%% file: iteration_main.tex
\documentclass[english,11pt, letter]{article}
\usepackage[T1]{fontenc}
\usepackage[latin9]{inputenc}
\usepackage{geometry}
\usepackage{array}
\usepackage{multirow}
\usepackage{amsthm}
\usepackage{amsmath}
\usepackage{amssymb}
\usepackage{esint}

\makeatletter

\providecommand{\tabularnewline}{\\}

\numberwithin{equation}{section}
\numberwithin{figure}{section}
\theoremstyle{plain}
\newtheorem{thm}{\protect\theoremname}
  \theoremstyle{definition}
  \newtheorem{defn}[thm]{\protect\definitionname}
  \theoremstyle{plain}
  \newtheorem{lem}[thm]{\protect\lemmaname}
  \theoremstyle{remark}
  \newtheorem{rem}[thm]{\protect\remarkname}
  \theoremstyle{plain}
  \newtheorem{cor}[thm]{\protect\corollaryname}

\pdfoutput=1

\usepackage{amsmath, amsthm, amssymb, amsfonts}
\usepackage{bbm}
\usepackage[basic]{complexity}
\usepackage[pdftex,bookmarks,colorlinks]{hyperref}
\usepackage{enumitem}
\usepackage[usenames,dvipsnames]{color}
\usepackage{tikz}

\usepackage[lined,boxed,ruled,norelsize,algo2e]{algorithm2e}

\DeclareMathOperator*{\argmaxTex}{arg\,max}
\DeclareMathOperator*{\argminTex}{arg\,min}
\DeclareMathOperator*{\signTex}{sign}
\DeclareMathOperator*{\rankTex}{rank}
\DeclareMathOperator*{\diagTex}{diag}
\DeclareMathOperator*{\imTex}{im}

\hypersetup{colorlinks=true,citecolor=red}

\renewcommand{\varepsilon}{\epsilon}

\makeatother

\usepackage{babel}
  \providecommand{\corollaryname}{Corollary}
  \providecommand{\definitionname}{Definition}
  \providecommand{\lemmaname}{Lemma}
  \providecommand{\remarkname}{Remark}
\providecommand{\theoremname}{Theorem}

\begin{document}
\include{iteration_notation}

\title{Efficient Inverse Maintenance\\
and Faster Algorithms for Linear Programming}

\author{Yin Tat Lee\\
MIT\\
yintat@mit.edu\and  Aaron Sidford\\
MIT\\
sidford@mit.edu}

\date{}
\maketitle
\begin{abstract}
In this paper, we consider the following \emph{inverse maintenance
problem:} given $\ma\in\R^{n\times d}$ and a number of rounds $r$,
at round $k$, we receive a $n\times n$ diagonal matrix $\md^{(k)}$
and we wish to maintain an efficient linear system solver for $\ma^{T}\md^{(k)}\ma$
under the assumption $\md^{(k)}$ does not change too rapidly. This
inverse maintenance problem is the computational bottleneck in solving
multiple optimization problems. We show how to solve this problem
with $\tilde{O}\left(\nnz(\ma)+d^{\omega}\right)$ preprocessing time
and amortized $\otilde(\nnz(\ma)+d^{2})$ time per round, improving
upon previous running times.

Consequently, we obtain the fastest known running times for solving
multiple problems including, linear programming and computing a rounding
of a polytope. In particular given a feasible point in a linear program
with $n$ variables, $d$ constraints, and constraint matrix $\ma\in\R^{d\times n}$,
we show how to solve the linear program in time $\otilde((\nnz(\ma)+d^{2})\sqrt{d}\log(\epsilon^{-1}))$.
We achieve our results through a novel combination of classic numerical
techniques of low rank update, preconditioning, and fast matrix multiplication
as well as recent work on subspace embeddings and spectral sparsification
that we hope will be of independent interest.
\end{abstract}
\input{sec_iteration_intro.tex}\input{sec_iteration_preliminaries.tex}

\input{sec_iteration_lowrank.tex}

\input{sec_iteration_sample.tex}

\input{sec_iteration_randomness.tex}

\input{sec_iteration_insertion_removal.tex}

\input{sec_iteration_applications.tex}

\section{Acknowledgments}

We thank Yan Kit Chim, Andreea Gane, and Jonathan A. Kelner for many
helpful conversations. This work was partially supported by NSF awards
0843915 and 1111109, NSF Graduate Research Fellowship (grant no. 1122374).
Part of this work was done while both authors were visiting the Simons
Institute for the Theory of Computing, UC Berkeley.

\bibliographystyle{plain}
\bibliography{main}

\appendix
\input{sec_iteration_appendix.tex}

\end{document}

%% file: iteration_notation.tex

\global\long\def\R{\mathbb{R}}
 \global\long\def\Rn{\mathbb{R}^{n}}
 \global\long\def\Rm{\mathbb{R}^{m}}
 \global\long\def\Rmn{\mathbb{R}^{m \times n}}
 \global\long\def\Rnm{\mathbb{R}^{n \times m}}
 \global\long\def\Rmm{\mathbb{R}^{m \times m}}
 \global\long\def\Rnn{\mathbb{R}^{n \times n}}
 \global\long\def\Z{\mathbb{Z}}
 \global\long\def\rPos{\R_{> 0}}
 \global\long\def\dom{\mathrm{dom}}
\global\long\def\dInterior{\Omega^{0}}
\global\long\def\dFull{\{\dInterior\times\rPos^{m}\}}
\global\long\def\dWeight{\rPos^{m}}


\global\long\def\ellOne{\ell_{1}}
 \global\long\def\ellTwo{\ell_{2}}
 \global\long\def\ellInf{\ell_{\infty}}
 \global\long\def\ellP{\ell_{p}}

\global\long\def\otilde{\tilde{O}}

\global\long\def\argmax{\argmaxTex}

\global\long\def\argmin{\argminTex}

\global\long\def\sign{\signTex}

\global\long\def\rank{\rankTex}

\global\long\def\diag{\diagTex}

\global\long\def\im{\imTex}

\global\long\def\enspace{\quad}

\global\long\def\boldVar#1{\mathbf{#1}}

\global\long\def\mvar#1{\boldVar{#1}}

\global\long\def\vvar#1{\vec{#1}}



\global\long\def\defeq{\stackrel{\mathrm{{\scriptscriptstyle def}}}{=}}

\global\long\def\diag{\mathrm{{diag}}}

\global\long\def\mDiag{\mvar{diag}}
 \global\long\def\ceil#1{\left\lceil #1 \right\rceil }

\global\long\def\E{\mathbb{E}}
 \global\long\def\abs#1{\left|#1\right|}

\global\long\def\gradient{\nabla}
\global\long\def\grad{\gradient}
 \global\long\def\hess{\gradient^{2}}
 \global\long\def\jacobian{\mvar J}
 \global\long\def\gradIvec#1{\vvar{f_{#1}}}
 \global\long\def\gradIval#1{f_{#1}}


\global\long\def\onesVec{\vec{1}}
 \global\long\def\indicVec#1{\onesVec_{#1}}
\global\long\def\indic{1}


\global\long\def\specGeq{\succeq}
 \global\long\def\specLeq{\preceq}
 \global\long\def\specGt{\succ}
 \global\long\def\specLt{\prec}

\global\long\def\va{\vvar a}
 \global\long\def\vb{\vvar b}
 \global\long\def\vc{\vvar c}
 \global\long\def\vd{\vvar d}
 \global\long\def\ve{\vvar e}
 \global\long\def\vf{\vvar f}
 \global\long\def\vg{\vvar g}
 \global\long\def\vh{\vvar h}
 \global\long\def\vl{\vvar l}
 \global\long\def\vm{\vvar m}
 \global\long\def\vn{\vvar n}
 \global\long\def\vo{\vvar o}
 \global\long\def\vp{\vvar p}
 \global\long\def\vs{\vvar s}
 \global\long\def\vu{\vvar u}
 \global\long\def\vv{\vvar v}
\global\long\def\vw{\vvar w}
 \global\long\def\vx{\vvar x}
 \global\long\def\vy{\vvar y}
 \global\long\def\vz{\vvar z}
 \global\long\def\vxi{\vvar{\xi}}
 \global\long\def\valpha{\vvar{\alpha}}
 \global\long\def\veta{\vvar{\eta}}
 \global\long\def\vphi{\vvar{\phi}}
\global\long\def\vpsi{\vvar{\psi}}
 \global\long\def\vsigma{\vvar{\sigma}}
 \global\long\def\vgamma{\vvar{\gamma}}
 \global\long\def\vphi{\vvar{\phi}}
\global\long\def\vdelta{\vvar{\delta}}
\global\long\def\vDelta{\vvar{\Delta}}
\global\long\def\vzero{\vvar 0}
 \global\long\def\vones{\vvar 1}

\global\long\def\ma{\mvar A}
 \global\long\def\mb{\mvar B}
 \global\long\def\mc{\mvar C}
 \global\long\def\md{\mvar D}
\global\long\def\mE{\mvar E}
 \global\long\def\mf{\mvar F}
 \global\long\def\mg{\mvar G}
 \global\long\def\mh{\mvar H}
 \global\long\def\mj{\mvar J}
 \global\long\def\mk{\mvar K}
 \global\long\def\mm{\mvar M}
 \global\long\def\mn{\mvar N}
 \global\long\def\mq{\mvar Q}
 \global\long\def\mr{\mvar R}
 \global\long\def\ms{\mvar S}
 \global\long\def\mt{\mvar T}
 \global\long\def\mU{\mvar U}
 \global\long\def\mv{\mvar V}
 \global\long\def\mx{\mvar X}
 \global\long\def\my{\mvar Y}
 \global\long\def\mz{\mvar Z}
 \global\long\def\mSigma{\mvar{\Sigma}}
 \global\long\def\mLambda{\mvar{\Lambda}}
\global\long\def\mPhi{\mvar{\Phi}}
\global\long\def\mpi{\mvar{\Pi}}
 \global\long\def\mZero{\mvar 0}
 \global\long\def\iMatrix{\mvar I}
\global\long\def\mDelta{\mvar{\Delta}}

\global\long\def\oracle{\mathcal{O}}

\global\long\def\runtime{\mathcal{T}}

\global\long\def\dWeights{\rPos^{m}}


\global\long\def\shurProd{\circ}
 \global\long\def\shurSquared#1{{#1}^{(2)}}

\global\long\def\weight{w}
 \global\long\def\vWeight{\vvar{\weight}}
 \global\long\def\mWeight{\mvar W}

\global\long\def\mProj{\mvar P}

\global\long\def\vLever{\vsigma}
 \global\long\def\mLever{\mSigma}
 \global\long\def\mLapProj{\mvar{\Lambda}}

\global\long\def\penalizedObjective{f_{t}}
 \global\long\def\penalizedObjectiveWeight{\hat{f}}

\global\long\def\fvWeight{\vg}
 \global\long\def\fmWeight{\mg}

\global\long\def\gradW{\grad_{\vWeight}}
 \global\long\def\gradX{\grad_{\vx}}
 \global\long\def\hessWW{\hess_{\vWeight\vWeight}}
 \global\long\def\hessWX{\hess_{\vWeight\vx}}
 \global\long\def\hessXW{\hess_{\vx\vWeight}}
 \global\long\def\hessXX{\hess_{\vx\vx}}

\global\long\def\vNewtonStep{\vh}


\global\long\def\norm#1{\big\|#1\big\|}
 \global\long\def\normFull#1{\left\Vert #1\right\Vert }
 \global\long\def\normA#1{\norm{#1}_{\ma}}
 \global\long\def\normFullInf#1{\normFull{#1}_{\infty}}
 \global\long\def\normInf#1{\norm{#1}_{\infty}}
 \global\long\def\normOne#1{\norm{#1}_{1}}
 \global\long\def\normTwo#1{\norm{#1}_{2}}
 \global\long\def\normLeverage#1{\norm{#1}_{\mSigma}}
 \global\long\def\normWeight#1{\norm{#1}_{\fmWeight}}

\global\long\def\cWeightSize{c_{1}}
 \global\long\def\cWeightStab{c_{\gamma}}
 \global\long\def\cWeightCons{c_{\delta}}

\global\long\def\TODO#1{{\color{red}TODO:\text{#1}}}

\global\long\def\mixedNorm#1#2{\normFull{#1}_{#2+\infty}}
\global\long\def\CNorm{C_{\mathrm{norm}}}
\global\long\def\Pxw{\mProj_{\vx,\vWeight}}
\global\long\def\vq{\vec{q}}
\global\long\def\cnorm{C_{\text{norm}}}

\global\long\def\next#1{#1^{\text{(new)}}}

\global\long\def\trInit{\vx^{(0)}}
 \global\long\def\trCurr{\vx^{(k)}}
 \global\long\def\trNext{\vx^{(k + 1)}}
 \global\long\def\trAdve{\vy^{(k)}}
 \global\long\def\trAfterAdve{\vy}
 \global\long\def\trMeas{\vz^{(k)}}
 \global\long\def\trAfterMeas{\vz}
 \global\long\def\trGradCurr{\grad\Phi_{\alpha}(\trCurr)}
 \global\long\def\trGradAdve{\grad\Phi_{\alpha}(\trAdve)}
 \global\long\def\trGradMeas{\grad\Phi_{\alpha}(\trMeas)}
 \global\long\def\trGradAfterAdve{\grad\Phi_{\alpha}(\trAfterAdve)}
 \global\long\def\trGradAfterMeas{\grad\Phi_{\alpha}(\trAfterMeas)}
 \global\long\def\trSetCurr{U^{(k)}}
\global\long\def\vWeightError{\vvar{\Psi}}
\global\long\def\code#1{\texttt{#1}}

\global\long\def\nnz{\mathrm{nnz}}
\global\long\def\solver{\mathrm{\mathtt{S}}}
\global\long\def\time{\mathrm{\mathcal{T}}}
\global\long\def\sampleConstant{c_{\mathrm{sample}}}

\newcommand{\bracket}[1]{[#1]}

%% file: sec_iteration_intro.tex
\section{Introduction}

Solving a sequence of linear systems is the computational bottleneck
in many state-of-the-art optimization algorithms, including interior
point methods for linear programming \cite{karmarkar1984new,Nesterov1994,renegar1988polynomial,leeS14},
the Dikin walk for sampling a point in a polytope \cite{kannan2012random},
and multiplicative weight algorithms for grouped least squares problem
\cite{chin2013runtime}, etc. In full generality, any particular iteration
of these algorithms could require solving an arbitrary positive definite
(PD) linear system. However, the PD matrices involved in these algorithms
do not change too much between iterations and therefore the average
cost per iteration can possibly be improved by maintaining an approximate
inverse for the matrices involved. 

This insight has been leveraged extensively in the field of interior
point methods for linear programming. Combining this insight with
classic numerical machinery including preconditioning, low ranks update,
and fast matrix multiplication has lead to multiple non-trivial improvements
to the state-of-the art running time for linear programming \cite{nesterov1991acceleration,Nesterov1994,vaidya1989speeding,leeS14,karmarkar1984new}.
Prior to our previous work \cite{leeS14}, the fastest algorithm for
solving a general linear program with $d$ variables, $n$ constraints,
and constraint matrix $\ma\in\R^{n\times d}$, i.e. solving $\min_{\ma\vx\geq\vb}\vc^{T}\vx$,
depended intricately on the precise ratio of $d$, $n$, and $\nnz(\ma)$
the number of non-zero entries in $\ma$ (See In Figure~\ref{fig:run_time}).
While, these running times were recently improved by our work in \cite{leeS14},
the running time we achieved was still a complicated bound of $\otilde(\sqrt{n\beta}(d^{2}+nd^{\omega-1}r^{-1}+\beta^{-\omega}r^{2\omega}+\beta^{-(\omega-1)}dr^{\omega}))$
and $\tilde{O}(\sqrt{d}(\nnz(\ma)+d^{\omega}))$%
\footnote{In this paper we use $\otilde$ to hide factors polylogarithmic in
$d$, $n$ and the width $U$. See Sec~\ref{sec:applications} for
the definition of $U$.%
} to compute an $\epsilon$-approximate solution where $\beta\in[\frac{d}{n},1]$
and $r\geq1$ are free parameters to be tuned and $\omega<2.3729$
is the matrix multiplication constant \cite{williams2012matrixmult,gall2014powers}.

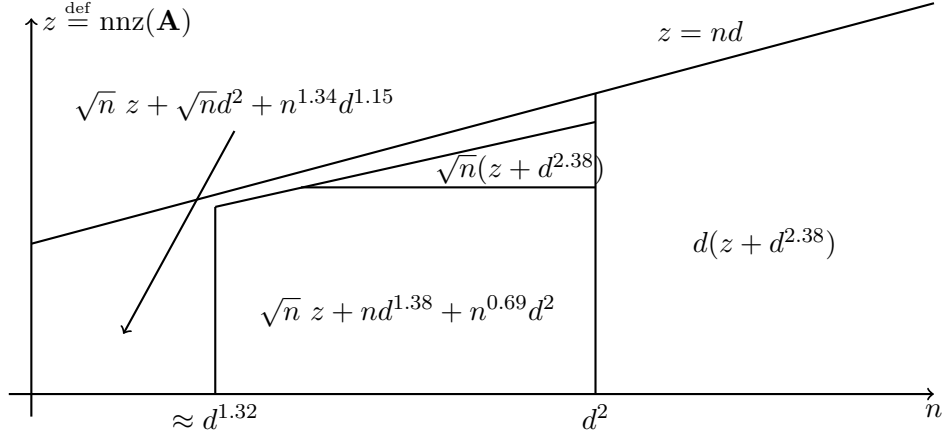
\begin{figure}
\begin{centering}
\begin{tikzpicture}[thick]   \draw[->] (-0.3,0) -- (12,0) coordinate[label = {below:$n$}] (xmax);   \draw[->] (0,-0.3) -- (0,5) coordinate[label = {right:$z \defeq \mathrm{nnz}(\mvar{A})$}] (ymax);   \draw       (0,2) -- (12,5.2) node[pos=0.8, above left] {$z=nd$};   \draw       (2.445,0) -- (2.445,2.49) node[pos=0, below] {$\approx d^{1.32}$};   \draw       (7.5,3.62) -- (2.445,2.49);   \draw       (7.5,2.75) -- (3.585,2.75);   \draw       (7.5,0) -- (7.5,4) node[pos=0, below] {$d^{2}$};   \draw[->] (2.7,3.5) -- (1.2225,0.8) node[pos=0,above] {$\sqrt{n} \ z+\sqrt{n}d^2+n^{1.34}d^{1.15}$};   \draw       (6.5,3) node[align=center] {$\sqrt{n} (z+d^{2.38})$};   \draw       (5,1.1) node[align=center] {$\sqrt{n} \ z+n d^{1.38}+n^{0.69}d^2$};   \draw       (9.75,2) node[align=center] {$d (z+d^{2.38})$};  \pgflowlevelsynccm \end{tikzpicture}
\par\end{centering}

\protect\caption{\label{fig:run_time} The figure shows the fastest algorithms for
solving a linear program, $\min_{\protect\ma\protect\vx\geq\protect\vb}\protect\vc^{T}\protect\vx$,
before \cite{leeS14}. The horizontal axis is the number of constraints
$n$ written as a function of $d$, the vertical access is the number
of nonzero entries in $\protect\ma$ denoted by $z$. Each region
of the diagram is the previous best running time for obtaining one
bit of precision in solving a linear program with the given parameters.
For more information on how these times were computed, see Appendix~\ref{sec:figure_remarks}.}
\end{figure}

In this paper we further improve upon these results. We cast this
computational bottleneck of solving a sequence of slowly changing
linear systems as a self-contained problem we call the \emph{inverse
maintenance problem }and improve upon the previous best running times
for solving this problem. This yields an improved running time of
$\otilde((\nnz(\ma)+d^{2})\sqrt{d}\log(\epsilon^{-1}$)) for solving
a linear program and improves upon the running time of various problems
including multicommodity flow and computing rounding of an ellipse.
We achieve our result by a novel application of recent numerical machinery
such as subspace embeddings and spectral sparsification to classic
techniques for solving this problem which we hope will be of independent
interest.

\subsection{The Inverse Maintenance Problem \label{sub:inverse_maintenance}}

The computational bottleneck we address in this paper is as follows.
We are given a fixed matrix $\ma\in\R^{n\times d}$ and an number
of rounds $r$. In each round $k\in[r]$ we are given a non-negative
vector $\vd^{(k)}\in\R_{\geq0}^{n}$ and we wish to maintain access
to an approximate inverse for the matrix $(\ma^{T}\md^{(k)}\ma)^{-1}$
where $\md^{(k)}\in\R^{n\times n}$ is a diagonal matrix with $\md_{ii}^{(k)}=d_{i}^{(k)}$.
We call this problem the \emph{inverse maintenance problem} and wish
to keep the cost of both maintaining the approximate inverse and applying
it as low as possible under certain stability assumptions on $\md^{(k)}$. 

For our applications we do not require that the approximate inverse
to be stored explicitly. Rather, we simply need an algorithm to solve
the linear system $\ma^{T}\md^{(k)}\ma\vx=\vb$ efficiently for any
input $\vb\in\R^{d}$. We formally define the requirements of this
\emph{solver} and define the\emph{ inverse maintenance problem} below.%
\footnote{In Appendix~\ref{sec:iteration_insertion_removal} we show that the
problem of maintaining a linear solver and maintaining an approximate
inverse are nearly equivalent.%
}
\begin{defn}[Linear System Solver]
 An algorithm $\solver$ is a $\mathcal{T}$-time solver of a PD
matrix $\mm\in\R^{d\times d}$ if for all $\vb\in\R^{d}$ and $\epsilon\in(0,1/2]$,
the algorithm outputs a vector $\solver(\vb,\epsilon)\in\R^{d}$ in
time $O(\time\log(\epsilon^{-1}))$ such that with high probability
in $n$, $\norm{\mathcal{\solver}(\vb,\epsilon)-\mm^{-1}\vb}_{\mm}^{2}\leq\epsilon\norm{\mm^{-1}\vb}_{\mm}^{2}$.
We call the algorithm $\solver$ linear if $\solver(\vb,\epsilon)=\mq_{\epsilon}\vb$
for some $\mq_{\varepsilon}\in\R^{d\times d}$ that depends only on
$\mm$ and $\epsilon$.
\begin{defn}[Inverse Maintenance Problem]
\label{def:problem} We are given a matrix $\ma\in\R^{n\times d}$
and a number of rounds $r>0$. In each round $k\in[r]$ we receive
$\vd^{(k)}\in\R_{>0}^{n}$ and wish to find a $\time$-time solver
$\solver^{(k)}$ for PD $\ma^{T}\md^{(k)}\ma$ such that both $\time$
and the cost of constructing $\solver^{(k)}$ is small. (See Algorithm~\ref{alg:maintenance}.)
\end{defn}
\end{defn}
\begin{algorithm2e}
\caption{Inverse Maintenance Problem}

\label{alg:maintenance}

\SetAlgoLined

\textbf{Input}: Full rank matrix $\ma\in\R^{n\times d}$, initial
scaling $\vd^{(0)}\in\R_{\geq0}^{n}$, the number of rounds $r>0$.\\
\,

\For{Each round $k\in[r]$}{

\textbf{Input}: Current point $\vd^{(k)}\in\R_{>0}^{n}$

\textbf{Output}: Solver $\mathcal{\solver}^{(k)}$ for $\ma^{T}\md^{(k)}\ma$
where $\md^{(k)}\in\R^{n\times n}$ is diagonal with $\md_{ii}^{(k)}=d_{i}^{(k)}$

}

\end{algorithm2e}

The inverse maintenance problem directly encompasses a computational
bottleneck in many interior point methods for linear programming \cite{lsInteriorPoint,lsMaxflow,renegar1988polynomial}
as well as similar methods for sampling a point in a polytope and
computing approximate John Ellipsoids. (See Section~\ref{sec:applications}
and \ref{sub:app:sample}.) Without further assumptions, we would
be forced to pay the cost of solving a general system, i.e. $\otilde(\nnz(\ma)+d^{\omega})$
\cite{nelson2012osnap,li2012iterative,Cohen2014} where $\omega<2.3729$
is the matrix multiplication constant \cite{williams2012matrixmult,gall2014powers}.
However, in these algorithms the $\md^{(k)}$ cannot change arbitrarily
and we develop algorithms to exploit this property. 

We consider two different assumptions on how $\md^{(k)}$ can change.
The first assumption, we call the $\ellTwo$ stability assumption,
and is met in most short-step interior point methods and many of our
applications. Under this assumption, the multiplicative change from
$\vd^{(k)}$ to $\vd^{(k+1)}$ is bounded in $\ellTwo$ norm, indicating
that multiplicative changes to coordinates happen infrequently on
average. The second assumption, we call the $\sigma$ stability assumption,
is weaker and holds both for these algorithms as well as a new interior
point method we provided recently in \cite{leeS14} to improve the
running time for linear programming. Under this assumption, the multiplicative
change from $\vd^{(k)}$ to $\vd^{(k+1)}$ is bounded in a norm induced
by the leverage scores, $\vsigma_{\ma}(\vd^{(k)})$, a common quantity
used to measure the importance of rows of a matrix (See Section~\ref{sub:leverage_scores}).
Here, multiplicative changes to these important rows happen infrequently
on average (See Section~\ref{sec:sigma_solution}). 
\begin{defn}[$\ellTwo$ Stability Assumption]
\label{ass:stability_l2} The inverse maintenance problem satisfies
the \emph{$\ellTwo$ stability assumption }if for all $k\in[r]$ we
have $\norm{\log(\vd^{(k)})-\log(\vd^{(k-1)})}_{2}\leq0.1$ and $\beta^{-1}\ma^{T}\md^{(0)}\ma\preceq\ma^{T}\md^{(k)}\ma\preceq\beta\ma^{T}\md^{(0)}\ma$
for $\beta=\poly(n)$.
\end{defn}

\begin{defn}[$\sigma$ Stability Assumption]
\label{ass:stability_sigma} We say that the inverse maintenance
problem satisfies the $\sigma$ stability assumption if for each $k\in[r]$
we have $\norm{\log(\vd^{(k)})-\log(\vd^{(k-1)})}_{\vLever_{\ma}(\vd^{(k)})}\leq0.1$,
$\norm{\log(\vd^{(k)})-\log(\vd^{(k-1)})}_{\infty}\leq0.1$, and $\beta^{-1}\ma^{T}\md^{(0)}\ma\preceq\ma^{T}\md^{(k)}\ma\preceq\beta\ma^{T}\md^{(0)}\ma$
for $\beta=\poly(n).$
\end{defn}
Note that many inverse maintenance algorithms do not require the condition
$\beta^{-1}\ma^{T}\md^{(0)}\ma\preceq\ma^{T}\md^{(k)}\ma\preceq\beta\ma^{T}\md^{(0)}\ma$;
However, this is a mild technical assumption which holds for all applications
we are interested in and it allows our algorithms to achieve further
running time improvements. 

Also note that the constant $0.1$ in these definitions is arbitrary.
If $\vd^{(k)}$ changes more quickly than this, then in many cases
we can simply add intermediate $\vd^{(k)}$ to make the conditions
hold and only changes the number of rounds by a constant. 

Finally, note that the $\sigma$ stability assumption is strictly
weaker the $\ellTwo$ stability assumption as $\sigma_{i}\leq1$ (See
Section~\ref{sub:leverage_scores}). We use $\ell^{2}$ assumption
mainly for comparison to previous works and as a warm up case for
our paper. In this paper, our central goal is to provide efficient
algorithms for solving the inverse maintenance under the weaker $\sigma$
stability assumption.

\subsection{Previous Work}

Work on the inverse maintenance problem dates back to the dawn of
interior point methods, a broad class of iterative methods for optimization.
In 1984, in the first proof of an interior point method solving a
linear program in polynomial time, Karmarkar \cite{karmarkar1984new}
observed that his algorithm needed to solve the inverse maintenance
problem under the $\ellTwo$ stability guarantee. Under this assumption
if one entry of $\vd^{(k)}$ changes by a constant then the rest barely
change at all. Therefore the updates to $\ma^{T}\md^{(k)}\ma$ are
essentially \emph{low rank}, i.e. on average they are well approximated
by the addition of a rank 1 matrix. He noted that it sufficed to maintain
an approximation $(\ma^{T}\md^{(k)}\ma)^{-1}$ and by using explicit
formulas for rank 1 updates he achieved an average cost of $\tilde{O}(n^{(\omega-2)/2}d^{2})$
for the inverse maintenance problem. This improved upon $\tilde{O}(nd^{\omega-1})$,
the best running time for solving the necessary linear system known
at that time.

\begin{figure}
\begin{centering}
\begin{tabular}{|c|l|c|}
\hline 
Year  & Author  & Amortized Cost Per Iteration\tabularnewline
\hline 
\hline 
\multicolumn{2}{|c|}{Best known linear system solver \cite{nelson2012osnap,li2012iterative,Cohen2014}} & $\tilde{O}(nnz(\ma)+d^{2.3729})$ \tabularnewline
\hline 
1984  & Karmarkar \cite{karmarkar1984new}  & $\tilde{O}(nnz(\ma)+n^{0.5}d^{1.3729}+n^{0.1865}d^{2})$\tabularnewline
\hline 
\multirow{1}{*}{1989} & \multirow{1}{*}{Nesterov and Nemirovskii \cite[Thm 8.4.1]{nesterov1989self,nesterov1991acceleration}} & $\tilde{O}(n^{0.5}d^{1.3729}+n^{0.1718}d^{1.9216}+n^{1.8987})$\tabularnewline
\hline 
1989 & Vaidya \cite{vaidya1989speeding}  & $O(nnz(\ma)+d^{2}+n^{0.8334}d^{1.1441})$\tabularnewline
\hline 
2014 & Lee and Sidford \cite{lsInteriorPoint} & $\tilde{O}(nnz(\ma)+d^{2}+n^{0.8260}d^{1.1340})$\tabularnewline
\hline 
2015 & This paper & $\tilde{O}(nnz(\ma)+d^{2})$\tabularnewline
\hline 
\end{tabular}
\par\end{centering}

\protect\caption{\label{fig:previous_runtime} Algorithms for solving the inverse maintenance
problem under $\ell^{2}$ stability guarantee. To simplify comparison
we let the \emph{amortized cost per iteration }denote the worst of
the average per iteration cost in maintaining the $\protect\time$-time
solver and $\protect\time$.}
\end{figure}

Nesterov and Nemirovskii \cite{nesterov1989self} built on the ideas
of Karmarkar. They showed that how to maintain an even cruder approximation
to $(\ma^{T}\md\ma)^{-1}$ and use it as a preconditioner for the
conjugate gradient method to solve the necessary linear systems. By
using the same rank 1 update formulas as Karmarkar and analyzing the
quality of these preconditioned systems they were able to improve
the running time for solving the inverse maintenance problem.

Vaidya \cite{vaidya1989speeding} noticed that instead of maintaining
an approximation to $(\ma^{T}\md\ma)^{-1}$ explicitly one can maintain
the inverse implicitly and group the updates together into blocks.
Using more general low-rank update formulas and fast matrix multiplication
he achieved an improved cost of $O(\nnz(\ma)+d^{2}+(nd^{(\omega-1)})^{5/6})$.
His analysis was tailored to the case where the matrix $\ma$ is dense;
his running time is $O(nd)$ which is essentially optimal when $\nnz(\ma)=nd$.
Focusing on the sparse regime, we showed that these terms that were
``lower order'' in his analysis can be improved by a more careful
application of matrix multiplication \cite{leeS14}. 

Note that, for a broad setting of parameters the previous fastest
algorithm for solving the inverse maintenance problem was achieved
by solving each linear system from scratch using fast regression algorithms
\cite{nelson2012osnap,li2012iterative,Cohen2014} These algorithms
all provide an $\tilde{O}(\nnz(\ma)+d^{\omega})$-time solver for
$\ma^{T}\md^{(k)}\ma$ directly by using various forms of dimension
reduction. The algorithm in \cite{nelson2012osnap} directly projected
the matrix to a smaller matrix and the algorithms in \cite{li2012iterative,Cohen2014}
each provide iterative algorithms for sampling rows from the matrix
to produce a smaller matrix.

Also note that we are unaware of a previous algorithm working directly
with the $\sigma$ stability assumption. Under this assumption it
is possible that many $\md^{(k)}$ change by a multiplicative constant
in a single round and therefore updates to $\ma^{T}\md^{(k)}\ma$
are no longer ``low rank'' on average, making low rank update techniques
difficult to apply cheaply. This is not just an artifact of an overly
weak assumption; our linear programming algorithm in \cite{leeS14}
had this same complication. In \cite{leeS14} rather than reasoning
about the $\sigma$ stability assumption we simply showed how to trade-off
the $\ellTwo$ stability of the linear systems with the convergence
of our algorithm to achieve the previous best running times.

\subsection{Our Approach}

\label{sub:approach}

In this paper we show how to solve the inverse maintenance under both
the $\ellTwo$ stability assumption and the weaker $\sigma$ stability
assumption such that the amortized maintenance cost per iteration
$\otilde(\nnz(\ma)+d^{2})$ and the solver runs in time $\otilde((\nnz(\ma)+d^{2})\log(\varepsilon^{-1}))$. 

To achieve our results, we show how to use low rank updates to maintain
a spectral sparsifier for $\ma^{T}\md^{(k)}\ma$, that is a weighted
subset of $\otilde(d)$ rows of $\ma$ that well approximate $\ma^{T}\md^{(k)}\ma$.
We use the well known fact that sampling the rows by leverage scores
(see Section~\ref{sub:leverage_scores}) provides precisely such
a guarantee and show that we can decrease both the number of updates
that need to be performed as well as the cost of those updates.

There are two difficulties that need to be overcome in this approach.
The first difficulty is achieving any running time that improves upon
both the low rank update techniques and the dimension reduction techniques;
even obtaining a running time of $\otilde(\nnz(\ma)+d^{\omega-c})$
for the current $\omega$ and $c>0$ for the inverse maintenance problem
under the $\ellTwo$ stability assumption was not known. Simply performing
rank 1 updates in each iteration is prohibitively expensive as the
cost of such an update is $O(d^{2})$ and there would be $\Omega(d^{c})$
such updates that need to be performed on average in each iteration
for some $c>0$. Furthermore, while Vaidya \cite{vaidya1989speeding}
overcame this problem by grouping the updates together and using fast
matrix multiplication, his approach needed to preprocess the rows
of $\ma$ by exactly computing $(\ma^{T}\md^{(0)}\ma)^{-1}\ma$. This
takes time $O(nd^{\omega-1})$ and is prohibitively expensive for
our purposes if $n$ is much larger than $d$.

To overcome this difficulty, we first compute an initial sparsifier
of $\ma^{T}\md^{(0)}\ma$ using only $\otilde(d)$ rows of $\ma$
in time $\otilde(\nnz(\ma)+d^{\omega})$ \cite{nelson2012osnap,li2012iterative,Cohen2014}.
Having performed this preprocessing, we treat low rank updates to
these original rows in the sparsifier and new rows separately. For
the original rows we can perform the preprocessing as in Vaidya \cite{vaidya1989speeding}
and the techniques in \cite{vaidya1989speeding,lsInteriorPoint} to
obtain the desired time bounds. For low rank updates to the new rows,
we show how to use subspace embeddings \cite{nelson2012osnap} to
approximate their contribution to sufficient accuracy without hurting
the asymptotic running times. In Section~\ref{sec:low_rank} we show
that this technique alone (even without use of sparsification), suffices
to mildly improve the running time of solving the inverse maintenance
assumption with the $\ellTwo$ stability assumption (but not to obtain
the fastest running times in this paper). 

The second difficulty is how to use the $\sigma$ stability assumption
to bound the number of updates to our sparsifier; under this weak
assumption where there can be many low rank changes to $\ma^{T}\md^{(k)}\ma$.
Here we simply show that the $\sigma$ stability assumption implies
that the leverage scores do not change too quickly in the norm induced
by leverage scores (See Section~\ref{sec:sigma_solution}). Consequently,
if we sample rows by leverage scores and only re-sample when either
$d_{i}^{(k)}$ or the leverage score changes sufficiently then we
can show that the expected number of low rank updates to our sparsifier
is small. This nearly yields the result, except that our algorithm
will use its own output to compute the approximate leverage scores
and the user of our algorithm might use the solvers to adversarially
decide the next $\vd^{(k)}$ and this would break our analysis. In
Section~\ref{sec:randomness} we show how to make our solvers indistinguishable
from a true solver plus fixed noise with high probability, alleviating
the issue.

Ultimately this yields amortized cost per iteration of $\otilde(\nnz(\ma)+d^{2})$
for the inverse maintenance problem. We remark that our algorithm
is one of few instances we can think of where an algorithm needs to
use both subspace embeddings \cite{nelson2012osnap} as well as iterative
linear system solving and sampling techniques \cite{Cohen2014,li2012iterative}
for different purposes to achieve a goal; for many applications they
are interchangeable.

In Section~\ref{sec:iteration_insertion_removal} we also show that
our approach can be generalized to accommodate for multiple new rows
to be added or removed from the matrix. We also show that our inverse
maintenance problem yields an $\otilde((\nnz(\ma)+d^{2})\sqrt{d}\log(\epsilon^{-1}))$
algorithm for solving a linear program improving upon our previous
work in \cite{leeS14} (See Section~\ref{sub:linear_programming}).
In Section~\ref{sec:applications}~we provide multiple applications
of our results to more specific problems including minimum cost flow,
multicommodity flow, non-linear regression, computing a rounding of
a polytope, and convex optimization and conclude with an open problem
regarding the sampling of a random point in a polytope (Section~\ref{sub:app:sample}).~

%% file: sec_iteration_preliminaries.tex
\section{Preliminaries\label{sec:preliminaries}}

\subsection{Notation\label{sub:notatin}}

Throughout this paper, we use vector notation, e.g $\vec{x}=(x_{1},\dots,x_{n})$,
to denote a vector and bold, e.g. $\ma$, to denote a matrix. We use
$\nnz(\vx)$ or $\nnz(\ma)$ to denote the number of nonzero entries
in a vector or a matrix respectively. Frequently, for $\vx\in\R^{d}$
we let $\mx\in\R^{d\times d}$ denote $\mDiag(\vx)$, the diagonal
matrix such that $\mx_{ii}=x_{i}$. For a symmetric matrix, $\mm$
and vector $\vx$ we let $\|\vx\|_{\mm}\defeq\sqrt{\vx^{T}\mm\vx}$.
For vectors $\vx$ and $\vy$ we let $\|\vy\|_{\vx}\defeq\sqrt{\sum_{i}x_{i}y_{i}^{2}}$.

In most applications, we use $d$ to denote the smaller dimension
of the problem, which can be either the number of variables or the
number of constraints. For example, the linear programs $\min_{\ma\vx=\vb,\vl\leq\vx\leq\vu}\vc^{T}\vx$
we solve in Theorem~\ref{thm:LPSolve} has $n$ variables and $d$
constraints (and use this formulation because it is more general than
$\min_{\ma\vx\geq\vb}\vc^{T}\vx$).

\subsection{Spectral Approximation\label{sub:approx}}

For symmetric matrices $\mn,\mm\in\R^{n\times n}$, we write $\mn\specLeq\mm$
to denote that $\vx^{T}\mn\vx\leq\vx^{T}\mm\vx$ for all $\vx\in\R^{n}$
and we define $\mn\specGeq\mm$, $\mn\prec\mm$ and $\mn\succ\mm$
analogously. We call a symmetric matrix $\mm$ positive definite (PD)
if $\mm\succ\mvar 0$. We use $\mm\approx_{\epsilon}\mn$ to denote
the condition that $e^{-\varepsilon}\mn\preceq\mm\preceq e^{\varepsilon}\mn$.

\subsection{Leverage Scores\label{sub:leverage_scores}}

Our algorithms make extensive use of \emph{leverage scores}, a common
measure of the importance of rows of a matrix. We denote the leverage
scores of a matrix $\ma\in\R^{n\times d}$ by $\vsigma\in\R^{n}$
and say the \emph{leverage score of row $i\in[n]$} is $\sigma_{i}\defeq[\ma\left(\ma^{T}\ma\right)^{-1}\ma^{T}]_{ii}$.
For $\ma\in\R^{n\times d}$, $\vd\in\rPos^{n}$, and $\md\defeq\mDiag(\vd)$
we use the shorthand $\vsigma_{\ma}(\vd)$ to denote the leverage
scores of the matrix $\md^{1/2}\ma$. We frequently use well known
facts regarding leverage scores, such as $\sigma_{i}\in[0,1]$ and
$\normOne{\vsigma}\leq d$. (See \cite{spielmanS08sparsRes,mahoney11survey,li2012iterative,Cohen2014}
for a more in-depth discussion of leverage scores, their properties,
and their many applications.)

We use the following two key results regarding leverage scores. The
first result states that one can sample $\tilde{O}(d)$ rows of $\ma$
according to their leverage score and obtain a spectral approximation
of $\ma^{T}\ma$ \cite{drineas2006subspace}. In Lemma \ref{lem:spectral_sparsification},
we use a variant of a random sampling lemma stated in \cite{Cohen2014}.
The second, Lemma~\ref{lem:computing_leverage_scores}, is a generalization
of a result in \cite{spielmanS08sparsRes} that has been proved in
various settings (see \cite{lsInteriorPoint} for example) that states
that given a solver for a $\ma^{T}\ma$ one can efficiently compute
approximate leverage scores for $\ma$. 
\begin{lem}[Leverage Score Sampling]
\label{lem:spectral_sparsification} Let $\ma\in\R^{n\times d}$
and let $\vu\in\R^{n}$ be an overestimate of $\vsigma$, the leverage
scores of $\ma$; i.e. $u_{i}\geq\sigma_{i}$ for all $i\in[n]$.
Set $p_{i}\defeq\min\left\{ 1,c_{s}\epsilon^{-2}u_{i}\log n\right\} $
for some absolute constant $c_{s}>0$ and $\epsilon\in(0,0.5]$ and
let $\mh\in\R^{n\times n}$ be a random diagonal matrix where independently
$\mh_{ii}=\frac{1}{p_{i}}$ with probability $p_{i}$ and is $0$
otherwise. With high probability in $n$, we have $\nnz(\mh)=O(\norm{\vu}_{1}\epsilon^{-2}\log n)$
and $\ma^{T}\mh\ma\approx_{\epsilon}\ma^{T}\ma$.
\end{lem}

\begin{lem}[Computing Leverage Scores]
\label{lem:computing_leverage_scores} Let $\ma\in\R^{n\times d}$,
let $\vsigma$ denote the leverage scores of $\ma$, and let $\epsilon>0$.
If we have a $\time$-time solver for $\ma^{T}\ma$ then in time $\otilde((\nnz(\ma)+\time)\epsilon^{-2}\log(\epsilon^{-1}))$
we can compute $\vec{\tau}\in\R^{n}$ such that with high probability
in $n$, $(1-\epsilon)\sigma_{i}\leq\tau_{i}\leq(1+\epsilon)\sigma_{i}$
for all $i\in[n]$.
\end{lem}

\subsection{Matrix Results}

Our algorithms combine the sampling techniques mentioned above with
the following two results.
\begin{lem}[Woodbury Matrix Identity]
\label{lem:Woodbury}For any matrices $\ma$, $\mathbf{U}$, $\mathbf{C}$,
$\mathbf{V}$ with compatible size, if $\ma$ and $\mathbf{C}$ are
invertible, then, $\left(\ma+\mathbf{U}\mathbf{C}\mathbf{V}\right)^{-1}=\ma^{-1}-\ma^{-1}\mathbf{U}\left(\mathbf{C}^{-1}+\mathbf{V}\ma^{-1}\mathbf{U}\right)^{-1}\mathbf{V}\ma^{-1}.$
\end{lem}

\begin{thm}[Theorem 9 in \cite{nelson2012osnap} ]
\label{thm:onsap}There is a distribution $\mathcal{D}$ over $\otilde(d\epsilon^{-2})\times n$
matrices such that for any $\ma\in\R^{n\times d}$, with high probability
in $n$ over the choice of $\mpi\sim\mathcal{D}$, we have $\ma^{T}\mpi^{T}\mpi\ma\approx_{\epsilon}\ma^{T}\ma$.
Sampling from $\mathcal{D}$ and computing $\mpi\ma$ for any $\mpi\in\mathrm{supp}(\mathcal{D})$
can be done in $\otilde(\nnz(\ma))$ time.\end{thm}

%% file: sec_iteration_lowrank.tex
\section{Solving the Inverse Maintenance Problem Using $\ell^{2}$ Stability\label{sec:low_rank}}

In this section we provide an efficient algorithm to solve the inverse
maintenance problem under the $\ellTwo$ stability assumption (See
Section~\ref{sub:inverse_maintenance}). The $\ellTwo$ stability
assumption is stronger than the $\sigma$ stability assumption and
the result we prove in this section is weaker than the one we prove
under the $\sigma$ stability assumption in the next section (although
still a mild improvement over many previous results). This section
serves as a ``warm-up'' to the more complicated results in Section~\ref{sec:sigma_solution}.

The main goal of this section is to prove the following theorem regarding
exactly maintaining an inverse under a sequence of low-rank updates.
We use this result for our strongest results on solving the inverse
maintenance problem under the $\sigma$ stability assumption in Section~\ref{sec:sigma_solution}.
\begin{thm}[Low Rank Inverse Maintenance]
\label{thm:low_rank_main} Suppose we are given a matrix $\ma\in\R^{n\times d}$,
a vector $\vb^{(0)}\in\R_{>0}^{d}$, a number of round $r>0$, and
in each round $k\in[r]$ we receive $\vd^{(k)}\in\R_{>0}^{n}$ such
that $\mb^{(k)}\defeq\ma^{T}\md^{(k)}\ma$ is PD. Further, suppose
that the number of pairs $(i,k)$ such that $d_{i}^{(k)}\neq d_{i}^{(k-1)}$
is bounded by $\alpha\leq d$ and suppose that there is $\beta>2$
such that $\beta^{-1}\mb^{(0)}\preceq\mb^{(k)}\preceq\beta\mb^{(0)}$.
Then, in round $k$, we can implicit construct a symmetric matrix
$\mk$ such that $\mk\approx_{O(1)}\left(\mb^{(k)}\right)^{-1}$ such
that we can apply $\mk$ to an arbitrary vector in time $\tilde{O}(d^{2}\log(\beta))$.
Furthermore, the algorithm in total takes time $\otilde(sd^{\omega-1}+\alpha rd\left(\frac{\alpha}{r}\right)^{\omega-2}+r\alpha^{\omega})$
where $s\defeq\max\{\nnz(\vd_{0}),d\}$.
\end{thm}
This improves upon the previous best expected running times in \cite{vaidya1989speeding,lsInteriorPoint}
which had an additive $\otilde(nd^{\omega-1})$ term which would be
prohibitively expensive for our purposes.

To motivate this theorem and our approach, in Section~\ref{sub:algorithm_l2_stable},
we prove that Theorem~\ref{thm:low_rank_main} suffices to yield
improved algorithms for the inverse maintenance problem under $\ellTwo$
stability. Then in Section~\ref{sub:low_rank_proof} we prove Theorem~\ref{thm:low_rank_main}
using a combination of classic machinery involving low rank update
formulas and new machinery involving subspace embeddings \cite{nelson2012osnap}.

\subsection{Inverse Maintenance under $\protect\ellTwo$ Stability\label{sub:algorithm_l2_stable}}

Here we show how Theorem~\ref{thm:low_rank_main} can be used to
solve the inverse maintenance problem under the $\ellTwo$ stability
assumption. Note this algorithm is primarily intended to illustrate
Theorem~\ref{thm:low_rank_main} and is a warm-up to the stronger
result in Section~\ref{sec:sigma_solution}.

\begin{algorithm2e}
\caption{Algorithm for the $\ell^2$ Stability Assumption}

\label{alg:sparseframework_strong_promise}

\SetAlgoLined

\textbf{Input: }Initial point $\vec{d}^{(0)}\in\R_{>0}^{n}$.

Set $\vd^{(apr)}:=\vec{d}^{(0)}$.

$\mq^{(0)}\defeq\ma^{T}\md^{(apr)}\ma$.

Let $\mk^{(0)}$ be an approximate inverse of $\mq^{(0)}$ computed
using Theorem~\ref{thm:low_rank_main}.

\textbf{Output:} A $\tilde{O}(d^{2}+\nnz(\ma))$-time linear solver
for $\ma^{T}\md^{(0)}\ma$ (using Lemma~\ref{lem:approx_implies_solver}
on $\mk^{(0)}$).

\For{each round $k\in[r]$}{

\textbf{Input:} Current point\textbf{ }$\vd^{(k)}\in\R_{>0}^{n}$.

\For{each coordinate $i\in[n]$}{

\uIf{ $0.9d_{i}^{(apr)}\leq d_{i}^{(k)}\leq1.1d_{i}^{(apr)}$ is
false}

{$d_{i}^{(apr)}:=d_{i}^{(k)}$.}}

$\mq^{(k)}\defeq\ma^{T}\md^{(apr)}\ma$.

Let $\mk^{(k)}$ be an approximate inverse of $\mq^{(k)}$ computed
using Theorem~\ref{thm:low_rank_main}.

\textbf{Output:} A $\tilde{O}(d^{2}+\nnz(\ma))$-time linear solver
for $\ma^{T}\md^{(k)}\ma$ (using Lemma~\ref{lem:approx_implies_solver}
on $\mk^{(k)}$).

}

\end{algorithm2e}
\begin{thm}
\label{thm:iteration_solver_strong_promise} Suppose that the inverse
maintenance problem satisfies the $\ell^{2}$ stability assumption.
Then Algorithm~\ref{alg:sparseframework_strong_promise} maintains
a $\tilde{O}(\nnz(\ma)+d^{2})$-time solver with high probability
in total time $\otilde(sd^{\omega-1}+r\left(nd^{\omega-1}\right)^{\frac{2\omega}{2\omega+1}})$
where $s=\max\{\max_{k\in[r]}\nnz(d^{(k)}),d\}$ and $r$ is the number
of rounds.\end{thm}
\begin{proof}
Recall that by the $\ell^{2}$ stability assumption $\norm{\log(\vd^{(k)})-\log(\vd^{(k-1)})}_{2}\leq0.1$
for all $k$. Therefore, in the $r$ rounds of the inverse maintenance
problem at most $O(r^{2})$ coordinates of $\vd^{(k)}$ change by
any fixed multiplicative constant. Consequently, the condition, $0.9d_{i}^{(apr)}\leq d_{i}^{(k)}\leq1.1d_{i}^{(apr)}$
in Algorithm~\ref{alg:sparseframework_strong_promise} is false at
most $O(r^{2})$ times during the course of the algorithm and at most
$O(r^{2})$ coordinates of the vector $\vd^{(apr)}$ change during
the course of the algorithm. 

Therefore, we can use Theorem~\ref{thm:low_rank_main} on $\ma^{T}\md^{(apr)}\ma$
with $\alpha=O(r^{2})$ and $s$ as defined in the theorem statement.
Consequently, the total cost of maintaining an implicit approximation
of $\left(\ma^{T}\md^{(apr)}\ma\right)^{-1}$ for $r$ rounds is $\tilde{O}\left(sd^{\omega-1}+dr^{\omega+1}+r^{2\omega+1}\right)$. 

To further improve the running time we restart the maintenance procedure
every $(sd^{\omega-1})^{\frac{1}{2\omega+1}}$ iterations. This yields
a total cost of maintenance that is bounded by 
\[
\tilde{O}\left(\left\lceil \frac{r}{(sd^{\omega-1})^{\frac{1}{2\omega+1}}}\right\rceil \left(sd^{\omega-1}+d\left((sd^{\omega-1})^{\frac{1}{2\omega+1}}\right)^{\omega+1}+\left((sd^{\omega-1})^{\frac{1}{2\omega+1}}\right)^{2\omega+1}\right)\right)\,.
\]
which is the same as 
\[
\tilde{O}\left(\left\lceil \frac{r}{(sd^{\omega-1})^{\frac{1}{2\omega+1}}}\right\rceil \left(sd^{\omega-1}+d(sd^{\omega-1})^{\frac{\omega+1}{2\omega+1}}\right)\right)\,.
\]
Since $\omega\leq1+\sqrt{2}$, we have $d(d^{\omega})^{\frac{\omega+1}{2\omega+1}}\leq d^{\omega}$
and thus $sd^{\omega-1}$ dominates $d(sd^{\omega-1})^{\frac{\omega+1}{2\omega+1}}$
when $s=d$. Furthermore, as $s\geq d$ the term $sd^{\omega-1}$
grows faster than $d(sd^{\omega-1})^{\frac{\omega+1}{2\omega+1}}$.
Consequently, $sd^{\omega-1}$ always dominates and we have the desired
result.
\end{proof}

\subsection{Low Rank Inverse Maintenance\label{sub:low_rank_proof}}

Here we prove Theorem~\ref{thm:low_rank_main} and provide an efficient
algorithm for maintaining the inverse of a matrix under a bounded
number of low rank modifications. The algorithm we present is heavily
motivated by the work in \cite{vaidya1989speeding} and the slight
simplifications in \cite{lsInteriorPoint}. However, our algorithm
improves upon the previous best cost of $\otilde(nd^{\omega-1}+dr^{\omega+1}+r^{2\omega+1})$
achieved in \cite{lsInteriorPoint} by a novel use of subspace embeddings
\cite{nelson2012osnap} that we hope will be of independent interest.%
\footnote{We require a slightly stronger assumption than \cite{lsInteriorPoint}
to achieve our result; \cite{lsInteriorPoint} did not require any
bound on $\beta$. %
} 

We break our proof of Theorem~\ref{thm:low_rank_main} into several
parts. First, we provide a simple technical lemma about maintaining
the weighted product of two matrices using fast matrix multiplication.
\begin{lem}
\label{lem:low_rank_helper} Let $\ma,\mb\in\R^{n\times d}$ and suppose
that in each round $k$ of $r$ we receive $\vx^{(k)},\vy^{(k)}\in\R^{n}$.
Suppose that the number of pairs $(i,k)$ such that $x_{i}^{(k)}\neq x_{i}^{(k-1)}$
or $y_{i}^{(k)}\neq y_{i}^{(k-1)}$ is upper bounded by $\alpha$.
Suppose that $\nnz(\vx^{(1)})\leq\alpha$, $\nnz(\vy^{(1)})\leq\alpha$
and $\alpha\leq d$. With $O(d\alpha^{\omega-1})$ time precomputation,
we can compute $\mx^{(k)}\ma\mb^{T}\my^{(k)}$ explicitly in an average
cost of $O(\alpha d\left(\frac{\alpha}{r}\right)^{\omega-2})$ per
iteration.\end{lem}
\begin{proof}
For the initial round, we compute $\mx^{(0)}\ma\mb^{T}\my^{(0)}$
explicitly by multiplying an $\alpha\times d$ and a $d\times\alpha$
matrix. Since $\alpha\leq d$, we can compute $\mx^{(0)}\ma\mb^{T}\my^{(0)}$
by multiplying $O\left(\frac{d}{\alpha}\right)$ matrices of size
$\alpha\times\alpha$. Using fast matrix multiplication this takes
time $O(d\alpha^{\omega-1})$. 

For all $k\in[r]$ let $\mDelta_{X}^{(k)}\defeq\mx^{(k)}-\mx^{(k-1)}$
and $\mDelta_{Y}^{(k)}=\my^{(k)}-\my^{(k-1)}$. Consequently, 
\[
\mx^{(k)}\ma\mb^{T}\my^{(k)}=\left(\mx^{(k-1)}+\mDelta_{X}^{(k)}\right)\ma\mb^{T}\left(\my^{(k-1)}+\mDelta_{Y}^{(k)}\right)\,.
\]
Note that $\nnz(\vx^{(k)})$ and $\nnz(\vy^{(k)})$ are less than
$2\alpha$. Thus, if we let $u_{k}$ denote the number of coordinates
$i$ such that $x_{i}^{(k)}\neq x_{i}^{(k-1)}$ or $y_{i}^{(k)}\neq y_{i}^{(k-1)}$,
then to compute $\mx^{(k)}\ma\mb^{T}\my^{(k)}$ we need only multiply
a $u_{k}\times d$ with a $d\times u_{k}$ matrix and multiply two
$2\alpha\times d$ matrices with $d\times u_{k}$ matrices. Since
$u_{k}\leq\alpha$, the running time is dominated by the time to compute
the $2\alpha\times d$ and a $d\times u_{k}$ . Since $u_{k}\leq\alpha\leq d$,
we can do this by multiplying $O\left(\frac{\alpha}{u_{k}}\cdot\frac{d}{u_{k}}\right)$
matrices of size $u_{k}\times u_{k}$. Using fast matrix multiplication,
this takes time $O(\alpha du_{k}^{\omega-2})$. 

Summing over all $u_{k}$ we see that the total cost of computing
the $\mx^{(k)}\ma\mb^{T}\my^{(k)}$ is
\[
O\left(\sum_{k=1}^{r}\alpha du_{k}^{\omega-2}\right)\leq O\left(\alpha dr\left(\frac{1}{r}\sum_{k=1}^{r}u_{k}\right)^{\omega-2}\right)\leq O\left(\alpha dr\left(\frac{\alpha}{r}\right)^{\omega-2}\right)\,,
\]
where in the second inequality we used the concavity of $x^{\omega-2}$.
Since this is at least the $O(d\alpha^{\omega-1})$ cost of computing
$\mx^{(0)}\ma\mb^{T}\my^{(0)}$ we obtain the desired result.
\end{proof}
Next, for completeness we prove a slightly more explicit variant of
a Lemma in \cite{lsInteriorPoint}.
\begin{lem}
\label{lem:low_rank_old} Let $\ma\in\R^{n\times d}$, $\vd^{(0)},\ldots,\vd^{(r)}\in\rPos^{n}$,
and $\mb^{(k)}\defeq\ma^{T}\md^{(k)}\ma$ for all $k$. Suppose that
the number of pairs $(i,k)$ such that $d_{i}^{(k)}\neq d_{i}^{(k-1)}$
is bounded by $\alpha$ and $\alpha\leq d$. In $O(nd^{\omega-1}+\alpha dr\left(\frac{\alpha}{r}\right)^{\omega-2})$
total time, we can explicitly output $\mc\in\R^{n\times d}$ and a
$\mv^{(k)}\in\R^{n\times n}$ in each round $k$ such that 
\[
\left(\mb^{(k)}\right)^{-1}=\left(\mb^{(0)}\right)^{-1}-\mc^{T}\Delta^{(k)}\mv^{(k)}\Delta^{(k)}\mc
\]
where $\Delta^{(k)}$ is a $n\times n$ diagonal matrix with $\Delta_{ii}^{(k)}=\vd_{i}^{(k)}-\vd_{i}^{(0)}$.
Furthermore, $\mv^{(k)}$ is an $\tilde{O}(\alpha)\times\tilde{O}(\alpha)$
matrix if we discard the zero rows and columns.\end{lem}
\begin{proof}
Let $\mb\defeq\mb^{(0)}$ for notational simplicity. Note that we
can compute $\mb$ directly in $O(nd^{\omega-1})$ time using fast
matrix multiplication. Then, using fast matrix multiplication we can
then compute $\mb^{-1}$ explicitly in $O(d^{\omega})$ time. Furthermore,
we can compute $\mc\defeq\ma\mb^{-1}$ in $O(nd^{\omega-1})$ time.

Let $u_{k}$ be the number of non-zero entries in $\Delta^{(k)}$
and $\mvar P^{(k)}$ be a $u_{k}\times n$ matrix such that the $\mvar P_{ij}^{(k)}=1$
if $j$ is the index of the $i^{th}$ non-zero diagonal entries in
$\Delta^{(k)}$ and $0$ otherwise. Let $\ms^{(k)}=\left(\mvar P^{(k)}\right)\left(\Delta^{(k)}\right)^{\dagger}\left(\mvar P^{(k)}\right)^{T}$.
By definition of $\ms^{(k)}$ and $\mvar P^{(k)}$, we have 
\[
\Delta^{(k)}=\Delta^{(k)}\left(\mvar P^{(k)}\right)^{T}\ms^{(k)}\mvar P^{(k)}\Delta^{(k)}.
\]
By the Woodbury matrix identity (Lemma~\ref{lem:Woodbury}), we know
that 
\begin{align*}
\left(\mb^{(k)}\right)^{-1} & =\left(\mb+\ma^{T}\Delta^{(k)}\left(\mvar P^{(k)}\right)^{T}\ms^{(k)}\mvar P^{(k)}\Delta^{(k)}\ma\right)^{-1}\\
 & =\mb^{-1}-\mc^{T}\Delta^{(k)}\left(\mvar P^{(k)}\right)^{T}\mvar T^{(k)}\mvar P^{(k)}\Delta^{(k)}\mc
\end{align*}
where $\mt^{(k)}=((\ms^{(k)})^{-1}+\mvar P^{(k)}\Delta^{(k)}\mc\ma^{T}\Delta^{(k)}(\mvar P^{(k)})^{T})^{-1}$.
Now by Lemma~\ref{lem:low_rank_helper}, we know that we can maintain
$\Delta^{(k)}\mc\ma^{T}\Delta^{(k)}$ in an additional $O(\alpha dr\left(\frac{\alpha}{r}\right)^{\omega-2})$
time, using that $\mc=\ma\mb^{-1}$ was precomputed in $O(nd^{\omega-1})$
time. Having the matrix $\Delta^{(k)}\mc\ma^{T}\Delta^{(k)}$ explicitly,
one can compute $\mvar P^{(k)}\Delta^{(k)}\mc\ma^{T}\Delta^{(k)}\left(\mvar P^{(k)}\right)^{T}$
in $O(u_{k}^{2})$ time and hence compute $\mt^{(k)}$ in $O(u_{k}^{\omega})$
time using fast matrix multiplication. Hence, the total running time
is $O(nd^{\omega-1}+\alpha dr\left(\frac{\alpha}{r}\right)^{\omega-2}+\sum u_{k}^{\omega})$.
The convexity of $x^{\omega}$ yields $\sum u_{k}^{\omega}\leq\alpha^{\omega}$.
Hence, the total time is bounded by
\[
O(nd^{\omega-1}+\alpha dr\left(\frac{\alpha}{r}\right)^{\omega-2}+\alpha^{\omega})=O(nd^{\omega-1}+\alpha dr\left(\frac{\alpha}{r}\right)^{\omega-2}).
\]

By setting $\mv^{(k)}=\left(\mvar P^{(k)}\right)^{T}\mt^{(k)}\mvar P^{(k)}$,
we have the desired formula. Note that $\mt^{(k)}$ is essentially
$\mv^{(k)}$ with the zero columns and rows and hence we have the
desired running time.
\end{proof}
We now everything we need to prove Theorem \ref{thm:low_rank_main}.
\begin{proof}[Proof of Theorem~\ref{thm:low_rank_main}]
Let $S\subset[n]$ denote the indices for which $\vd^{(0)}$ is nonzero.
Furthermore, let us split each vector $\vd^{(k)}$ into a vector $\ve^{(k)}$
for the coordinates in $S$ and $\vf^{(k)}$ for the coordinates not
in $S$, i.e. $\ve^{(k)},\vf^{(k)}\in\rPos^{n}$ such that $\vd^{(k)}=\ve^{(k)}+\vf^{(k)}$.
By the stability guarantee, we know $\beta^{-1}\mb^{(0)}\preceq\mb^{(k)}\preceq\beta\mb^{(0)}$
for some $\beta$. We make $\ma^{T}\mvar E^{(k)}\ma$ invertible for
all $k$ by adding $\frac{1}{10\beta}\vd^{(0)}$ to all of $\vd^{(k)}$and
$\ve^{(k)}$; we will show this only increases the error slightly.

Now, we can compute $\mb^{(0)}$ and $\left(\mb^{(0)}\right)^{-1}$
in $O(sd^{\omega-1})$ time using fast matrix multiplication where
$s=|S|$. Furthermore, using Lemma~\ref{lem:low_rank_old} we can
compute $\mc$ and maintain $\mv^{(k)}$ such that 
\[
\left(\ma^{T}\mvar E^{(k)}\ma\right)^{-1}=\left(\mb^{(1)}\right)^{-1}-\mc^{T}\Delta^{(k)}\mv^{(k)}\Delta^{(k)}\mc
\]
in total time $O(sd^{\omega-1}+\alpha dr\left(\frac{\alpha}{r}\right)^{\omega-2})$.
All that remains is to maintain the contribution from $\vf^{(k)}$. 

Using our representation of $\left(\ma^{T}\mvar E^{(k)}\ma\right)^{-1}$
and the Woodbury matrix identity (Lemma~\ref{lem:Woodbury}), we
have 
\begin{align}
\left(\mb^{(k)}\right)^{-1} & =\left(\ma^{T}\mE^{(k)}\ma+\ma^{T}\mf^{(k)}\ma\right)^{-1}\nonumber \\
 & =\left(\ma^{T}\mE^{(k)}\ma+\ma^{T}\mf^{(k)}\left(\mf^{(k)}\right)^{\dagger}\mf^{(k)}\ma\right)^{-1}\nonumber \\
 & =\left(\ma^{T}\mE^{(k)}\ma\right)^{-1}-\left(\ma^{T}\mE^{(k)}\ma\right)^{-1}\ma^{T}\mf^{(k)}\left(\mm^{(k)}\right)^{-1}\mf^{(k)}\ma\left(\ma^{T}\mE^{(k)}\ma\right)^{-1}\label{eq:B_eq}
\end{align}
where
\begin{align}
\mm^{(k)} & =\left(\left(\mf^{(k)}\right)^{\dagger}\right)^{-1}+\mf^{(k)}\ma\left(\ma^{T}\mE^{(k)}\ma\right)^{-1}\ma^{T}\mf^{(k)}\label{eq:M_eq}
\end{align}
Now note that 
\begin{align*}
\mf^{(k)}\ma\left(\ma^{T}\mE^{(k)}\ma\right)^{-1}\ma^{T}\mf^{(k)} & =\mf^{(k)}\ma\left(\ma^{T}\mE^{(k)}\ma\right)^{-1}\left(\ma^{T}\mE^{(k)}\ma\right)\left(\ma^{T}\mE^{(k)}\ma\right)^{-1}\ma^{T}\mf^{(k)}\\
 & =\left(\mn^{(k)}\right)^{T}\mn^{(k)}
\end{align*}
where 
\begin{align*}
\mn^{(k)} & \defeq\left(\mE^{(k)}\right)^{1/2}\ma\left(\ma^{T}\mE^{(k)}\ma\right)^{-1}\ma^{T}\mf^{(k)}\\
 & =\left(\left(\md^{(0)}\right)^{1/2}+\left(\left(\mE^{(k)}\right)^{1/2}-\left(\md^{(0)}\right)^{1/2}\right)\right)\ma\left(\left(\mb^{(0)}\right)^{-1}-\mc^{T}\Delta^{(k)}\mv^{(k)}\Delta^{(k)}\mc\right)\ma^{T}\mf^{(k)}.
\end{align*}

Note that computing the $\left(\md^{(0)}\right)^{1/2}\ma\left(\mb^{(0)}\right)^{-1}\ma^{T}\mf^{(k)}$
term directly would be prohibitively expensive. Instead, we compute
a spectral approximation to $\left(\mn^{(k)}\right)^{T}\mn^{(k)}$
and show that suffices. Since $\left(\mn^{(k)}\right)^{T}\mn^{(k)}$
is a rank $\alpha$ matrix, we use Theorem~\ref{thm:onsap} to $\mpi\in\R^{\otilde(\alpha)\times d}$
such that 
\begin{equation}
\left(\mn^{(k)}\right)^{T}\mpi^{T}\mpi\mn^{(k)}\approx_{1}\left(\mn^{(k)}\right)^{T}\mn^{(k)}\label{eq:n_pi_gurantee}
\end{equation}
 for all $k$ with high probability. Now, we instead consider the
cost of maintaining
\[
\left(\mn^{(k)}\right)^{T}\mpi^{T}\mpi\mn^{(k)}.
\]

To see the cost of maintaining $\mpi\mn^{(k)}$, we separate the matrix
as follows:
\begin{eqnarray}
\mpi\mn^{(k)} & = & \mpi\left(\md^{(0)}\right)^{1/2}\ma\left(\mb^{(0)}\right)^{-1}\ma^{T}\mf^{(k)}\label{eq:pi_n_terms}\\
 &  & +\mpi\left(\left(\mE^{(k)}\right)^{1/2}-\left(\md^{(0)}\right)^{1/2}\right)\ma\left(\mb^{(0)}\right)^{-1}\ma^{T}\mf^{(k)}\nonumber \\
 &  & -\mpi\left(\md^{(0)}\right)^{1/2}\ma\mc^{T}\Delta^{(k)}\mv^{(k)}\Delta^{(k)}\mc\ma^{T}\mf^{(k)}\nonumber \\
 &  & -\mpi\left(\left(\mE^{(k)}\right)^{1/2}-\left(\md^{(0)}\right)^{1/2}\right)\ma\mc^{T}\Delta^{(k)}\mv^{(k)}\Delta^{(k)}\mc\ma^{T}\mf^{(k)}.\nonumber 
\end{eqnarray}

For the first term in \eqref{eq:pi_n_terms}, note that $\left(\md^{(0)}\right)^{1/2}\ma$
is a $s\times d$ matrix and hence we can precompute $\left(\md^{(0)}\right)^{1/2}\ma\left(\mb^{(0)}\right)^{-1}$
in $\tilde{O}(sd^{\omega-1})$ time and therefore precompute $\mpi\left(\md^{(0)}\right)^{1/2}\ma\left(\mb^{(0)}\right)^{-1}$
in the same amount of time. Note that $\mpi$ is a $\tilde{O}(\alpha)\times d$
matrix, so, we can write
\[
\mpi\left(\md^{(0)}\right)^{1/2}\ma\left(\mb^{(0)}\right)^{-1}\ma^{T}\mf^{(k)}=\mLambda\mk\ma^{T}\mf^{(k)}
\]
where $\mk$ is some explicitly computed $n\times d$ matrix and $\mLambda$
is a diagonal matrix with only $\tilde{O}(\alpha)$ non-zeros. Consequently,
we can use Lemma~\ref{lem:low_rank_helper} to maintain $\mpi(\md^{(0)})^{1/2}\ma(\mb^{(0)})^{-1}\ma^{T}\mf^{(k)}$
in total time $\tilde{O}(sd^{\omega-1}+\alpha dr\left(\frac{\alpha}{r}\right)^{\omega-2}$).

For the second term in \eqref{eq:pi_n_terms}, we can precompute $\ma\left(\mb^{(0)}\right)^{-1}$
in $\tilde{O}(sd^{\omega-1})$ time. Therefore, using Lemma~\ref{lem:low_rank_helper}
we can maintain $\left(\left(\mE^{(k)}\right)^{1/2}-\left(\md^{(0)}\right)^{1/2}\right)\ma\left(\mb^{(0)}\right)^{-1}\ma^{T}\mf^{(k)}$
in total time $\tilde{O}\left(\alpha dr\left(\frac{\alpha}{r}\right)^{\omega-2}\right)$
and by Theorem~\ref{thm:onsap} we can maintain $\mpi\left(\left(\mE^{(k)}\right)^{1/2}-\left(\md^{(0)}\right)^{1/2}\right)\ma\left(\mb^{(0)}\right)^{-1}\ma^{T}\mf^{(k)}$
in the same amount of time.

For the last two terms in \eqref{eq:pi_n_terms}, we use Lemma~\ref{lem:low_rank_helper}
on $\Delta^{(k)}\mc\ma^{T}\mf^{(k)}$, $\mpi\left(\md^{(0)}\right)^{1/2}\ma\mc^{T}\Delta^{(k)}$
and $\left(\left(\mE^{(k)}\right)^{1/2}-\left(\md^{(0)}\right)^{1/2}\right)\ma\mc^{T}\Delta^{(k)}$
all in total time $\tilde{O}\left(\alpha dr\left(\frac{\alpha}{r}\right)^{\omega-2}\right)$.
Note that all of those matrices including $\mv^{(k)}$ are essentially
$\tilde{O}(\alpha)\times\tilde{O}(\alpha)$ matrices if we discard
the zero rows and columns. So, having those matrices explicitly computed,
we can compute the last two terms in an additional of $\tilde{O}(\alpha^{\omega})$
time per iteration.

Finally computing $\left(\mn^{(k)}\right)^{T}\mpi^{T}\mpi\mn^{(k)}$
only requires an additional $\otilde(\alpha^{\omega})$ time per-iteration.
Hence, the total cost of maintaining $\left(\mn^{(k)}\right)^{T}\mpi^{T}\mpi\mn^{(k)}$
is 
\[
\tilde{O}\left(\alpha dr\left(\frac{\alpha}{r}\right)^{\omega-2}+sd^{\omega-1}+r\alpha^{\omega}\right).
\]

Using \eqref{eq:B_eq} and \eqref{eq:M_eq}, we have shown how to
approximate $\left(\mb^{(k)}\right)^{-1}$. Now, we show how to compute
a better approximation. Recall from \eqref{eq:B_eq} that
\[
\left(\mb^{(k)}\right)^{-1}=\left(\ma^{T}\mE^{(k)}\ma\right)^{-1}+\left(\ma^{T}\mE^{(k)}\ma\right)^{-1}\ma^{T}\mf^{(k)}\left(\mm^{(k)}\right)^{-1}\mf^{(k)}\ma\left(\ma^{T}\mE^{(k)}\ma\right)^{-1}.
\]
Using Lemma~\ref{lem:low_rank_old} as we have argued, we can apply
the first term $\left(\ma^{T}\mE^{(k)}\ma\right)^{-1}$ exactly to
a vector in $\otilde(d^{2})$ time. The only difficulty in applying
the second term to a vector comes from the term $\left(\mm^{(k)}\right)^{-1}$.
Using \eqref{eq:n_pi_gurantee}, we have
\[
\mm^{(k)}\approx_{1}\left(\mf^{(k)}\right)^{\dagger}+\left(\mn^{(k)}\right)^{T}\mpi^{T}\mpi\mn^{(k)}
\]
and hence
\[
\left(\mm^{(k)}\right)^{-1}\approx_{1}\left(\left(\mf^{(k)}\right)^{\dagger}+\left(\mn^{(k)}\right)^{T}\mpi^{T}\mpi\mn^{(k)}\right)^{-1}.
\]
Since computing $\left(\left(\mf^{(k)}\right)^{\dagger}+\left(\mn^{(k)}\right)^{T}\mpi^{T}\mpi\mn^{(k)}\right)^{-1}$
only requires $\otilde(\alpha^{\omega})$ time and \eqref{eq:M_eq}
shows that we can apply $\mm^{(k)}$ to a vector exactly in $\otilde(d^{2})$
time, by Lemma \ref{lem:approx_implies_solver} we have a symmetric
matrix $\mm_{\epsilon}^{(k)}\approx_{\epsilon}\mm^{(k)}$ such that
we can apply $\left(\mm_{\epsilon}^{(k)}\right)^{-1}$ to a vector
in $\otilde(d^{2}\log(\epsilon^{-1}))$ for any $\epsilon>0$. Hence,
we obtain 
\[
\mLambda_{\epsilon}^{(k)}\approx_{\epsilon}\left(\ma^{T}\mE^{(k)}\ma\right)^{-1}\ma^{T}\mf^{(k)}\left(\mm^{(k)}\right)^{-1}\mf^{(k)}\ma\left(\ma^{T}\mE^{(k)}\ma\right)^{-1}
\]
such that we can apply $\mLambda_{\epsilon}^{(k)}$ to a vector in
$\otilde(d^{2}\log(\epsilon^{-1})$) time. All that remains is to
compute what value of $\epsilon$ is needed for this to be a spectral
approximation to the $\left(\mb^{(k)}\right)^{-1}$.

Recall that by the assumption $\beta^{-1}\mb^{(0)}\preceq\mb^{(k)}\preceq\beta\mb^{(0)}$.
As we mentioned in the beginning, we replaced $\vd^{(k)}$ with $\vd^{(k)}+\frac{1}{10\beta}\vd^{(0)}$
and compute a constant spectral approximation to the new $\mb^{(k)}$
that will suffice for the theorem statement. Consequently 
\[
\mLambda^{(k)}\defeq\left(\ma^{T}\mE^{(k)}\ma\right)^{-1}-\left(\mb^{(k)}\right)^{-1}\specLeq10\beta\left(\mb^{(0)}\right)^{-1}
\]
Furthermore, since $\vf^{(k)}\geq0$ we have $\mvar 0\preceq\mLambda^{(k)}$
and therefore
\[
\mLambda^{(k)}-10\epsilon\beta\left(\mb^{(0)}\right)^{-1}\specLeq\mLambda_{\varepsilon}^{(k)}\preceq\mLambda^{(k)}+10\epsilon\beta\left(\mb^{(0)}\right)^{-1}.
\]
Using the assumption $\beta^{-1}\mb^{(0)}\preceq\mb^{(k)}\preceq\beta\mb^{(0)}$
again, we have
\[
\mLambda^{(k)}-10\epsilon\beta^{2}\left(\mb^{(k)}\right)^{-1}\specLeq\mLambda_{\varepsilon}^{(k)}\preceq\mLambda^{(k)}+10\epsilon\beta^{2}\left(\mb^{(k)}\right)^{-1}.
\]
Picking $\epsilon=O(\frac{1}{20\beta^{2}})$, we have a matrix $\left(\ma^{T}\mE^{(k)}\ma\right)^{-1}-\mLambda_{\epsilon}^{(k)}\approx_{O(1)}\left(\mb^{(k)}\right)^{-1}$
which we can apply in $\otilde(d^{2}\log(\beta))$ time. Furthermore,
again since $\beta^{-1}\mb^{(0)}\preceq\mb^{(k)}\preceq\beta\mb^{(0)}$
we see that our replacement of d $\vd^{(k)}$ with $\vd^{(k)}+\frac{1}{10\beta}\vd^{(0)}$only
affected our approximation quality by a constant.\end{proof}

%% file: sec_iteration_sample.tex
\section{An algorithm for the $\sigma$ Stable Case \label{sec:sigma_solution}}

In this section, we provide our algorithm for solving the inverse
maintenance problem under the $\sigma$ stability assumption. The
central result of this section is the following:
\begin{thm}
\label{thm:iteration_solver} Suppose that the inverse maintenance
problem satisfies the $\sigma$ stability assumption. Then Algorithm~\ref{alg:sparseframework}
maintains a $\tilde{O}(\nnz(\ma)+d^{2})$-time solver with high probability
in total time $\otilde(d^{\omega}+r(\nnz(\ma)+d^{2}))$ where $r$
is the number of rounds.
\end{thm}
To prove this we first provide a technical lemma showing that leverage
scores are stable in leverage score number assuming $\sigma$ stability
(See Section~\ref{sec:leverage_score_stability}). Using this lemma
we then prove the Theorem~\ref{thm:iteration_solver} (See Section~\ref{sub:sigma_stable_algorithm}).
This proof will assume that the randomness we use to maintain our
solvers is independent from the output of our solvers. In Section~\ref{sec:randomness}
we show how to make this assumption hold.

\subsection{Leverage Scores are Stable under $\sigma$ Stability}

\label{sec:leverage_score_stability}

Here we show that leverage scores are stable in the leverage score
norm assuming $\sigma$ stability. This technical lemma, Lemma~\ref{lem:lever_score_continuity},
is crucial to showing that we do not need to perform too many low-rank
updates on our sparsifier in our solution to the inverse maintenance
problem under the $\sigma$ stability assumption. (See Section~\ref{sub:approach}
for more intuition.)
\begin{lem}
\label{lem:lever_score_continuity}For all $\ma\in\R^{n\times d}$
and any vectors $\vx,\vy\in\rPos^{n}$ such that $\norm{\ln(\vx)-\ln(\vy)}_{\infty}\leq\epsilon$,
we have
\[
\norm{\ln\vsigma_{\ma}(\vx)-\ln\vsigma_{\ma}(\vy)}_{\vsigma_{\ma}(\vx)}\leq e^{\epsilon}\norm{\ln\vx-\ln\vy}_{\vsigma_{\ma}(\vx)}\,.
\]
\end{lem}
\begin{proof}
For $t\in[0,1]$, let $\ln\vec{\theta}_{t}$ denote a straight line
from $\ln\vx$ to $\ln\vy$ with $\vec{\theta}_{0}=\vx$ and $\vec{\theta}_{1}=\vy$
or equivalently let $\vec{\theta}_{t}\defeq\exp(\ln x+t(\ln y-\ln x))$).
Since $\norm{\ln(\vx)-\ln(\vy)}_{\infty}\leq\epsilon$ we have for
all $i\in[n]$
\begin{align}
\vLever_{\ma}(\vx)_{i} & =\indicVec i^{T}\sqrt{\mx}\ma\left(\ma^{T}\mx\ma\right)^{-1}\ma^{T}\sqrt{\mx}\indicVec i\nonumber \\
 & \approx_{2\epsilon}\indicVec i^{T}\sqrt{\Theta}\ma\left(\ma^{T}\Theta\ma\right)^{-1}\ma^{T}\sqrt{\Theta}\indicVec i\nonumber \\
 & =\vsigma_{\ma}(\vec{\theta}_{i})\label{eq:sigma_inequality}
\end{align}
Consequently, for all $\vz$, we have $\norm{\vz}_{\vsigma_{\ma}(\vvar{\theta}_{t})}\leq e^{\epsilon}\norm{\vz}_{\vsigma_{\ma}}$
and by Jensen's inequality we have
\begin{equation}
\norm{\ln\vLever_{\ma}(\vx)-\ln\vLever_{\ma}(\vy)}_{\sigma_{\ma}(\vx)}\leq\normFull{\int_{0}^{1}\left(\frac{d}{dt}\ln\vLever_{\ma}(\vec{\theta}_{t})\right)dt}_{\sigma_{\ma}(\vx)}\leq e^{\epsilon}\cdot\int_{0}^{1}\normFull{\frac{d}{dt}\ln\vLever_{\ma}(\vec{\theta}_{t})}_{\sigma_{\ma}(\vvar{\theta}_{t})}dt\label{eq:sigma_proof_2}
\end{equation}
Now, for all $z\in\R_{>0}^{n}$ let $\jacobian_{\ln\vz}(\ln\vLever_{\ma}(\vz))$
denote the Jacobian matrix $\left[\frac{\partial\ln\vLever_{\ma}(\vz)_{i}}{\partial\ln\vz_{j}}\right]_{ij}=\left[\frac{\partial\ln\vLever_{\ma}(\vz)_{i}}{\partial\vz_{j}}z_{j}\right]_{ij}$
and suppose that for all $\vz$ and $\vu$ we have
\begin{equation}
\norm{\mj_{\ln\vz}(\ln\vsigma_{\ma}(\vz))\vu}_{\vsigma_{\ma}}\leq\norm{\vu}_{\vsigma_{\ma}(\vz)}.\label{eq:J_norm_1}
\end{equation}
Then
\begin{align}
\int_{0}^{1}\normFull{\frac{d}{dt}\ln\vLever_{\ma}(\vec{\theta}_{t})}_{\sigma_{\ma}(\vvar{\theta}_{t})}dt & =\int_{0}^{1}\norm{\mj_{\ln\vvar{\theta}_{t}}(\ln\vsigma_{\ma}(\vvar{\theta}_{t}))\left(\frac{d}{dt}\ln\vvar{\theta}_{t}\right)}_{\vsigma_{\ma}(\vvar{\theta}_{t})}\leq e^{\epsilon}\norm{\ln\vx-\ln\vy}_{\vsigma_{\ma}(\vx)}\label{eq:sigma_proof_3}
\end{align}
where again used \ref{eq:sigma_inequality} as well as the definition
of $\vvar{\theta}_{t}$. All that remains is to prove \eqref{eq:J_norm_1}.
For this, we first note that in \cite[Ver 1, Lemma 36]{lsInteriorPoint}
we showed that $\jacobian_{\vz}(\vLever_{\ma}(\vz))=\left(\Sigma_{\ma}(\vz)-\mm\right)\mz^{-1}$
where $\mm_{ij}(\vz)\defeq\left(\sqrt{\mz}\ma(\ma^{T}\mz\ma)^{-1}\ma^{T}\sqrt{\mz}\right)^{2}$.
Consequently $\jacobian_{\ln\vz}(\ln\vLever_{\ma}(\vz))=\Sigma_{\ma}(\vz)^{-1}\left(\Sigma_{\ma}(\vz)-\mm(\vz)\right)$.
Now note that
\begin{eqnarray*}
\sum_{i}\mm_{ij} & = & \sum_{i}\left(\sqrt{\mz}\ma(\ma^{T}\mz\ma)^{-1}\ma^{T}\sqrt{\mz}\right)_{ij}^{2}\\
 & = & \left\langle \sqrt{\mz}\ma(\ma^{T}\mz\ma)^{-1}\ma^{T}\sqrt{\mz}\onesVec_{j},\sqrt{\mz}\ma(\ma^{T}\mz\ma)^{-1}\ma^{T}\sqrt{\mz}\onesVec_{j}\right\rangle \\
 & = & \left(\sqrt{\mz}\ma(\ma^{T}\mz\ma)^{-1}\ma^{T}\sqrt{\mz}\right)_{jj}\\
 & = & \left(\vLever_{\ma}(\vz)\right)_{j}.
\end{eqnarray*}
Consequently, $\Sigma_{\ma}(\vz)-\mm$ is a symmetric diagonally dominant
matrix and therefore $\Sigma_{\ma}(\vz)\succeq\Sigma_{\ma}(\vz)-\mm\succeq\mZero$
and $\mZero\preceq\Sigma_{\ma}(\vz)^{1/2}\jacobian_{\ln\vz}(\ln\vLever_{\ma}(\vz))\Sigma_{\ma}(\vz)^{-1/2}\preceq\iMatrix$.
Using this fact, we have that for all $\vz\in\R_{>0}^{n}$ and $\vu\in\R^{n}$
\[
\norm{\jacobian_{\ln\vz}(\ln\vLever_{\ma}(\vz))\vu}_{\vsigma_{\ma}(\vz)}=\norm{\Sigma_{\ma}(\vz)^{1/2}\vu}_{\left(\Sigma_{\ma}(\vz)^{-1/2}\left(\Sigma_{\ma}(\vz)-\mm\right)\Sigma_{\ma}(\vz)^{-1/2}\right)^{2}}\leq\norm{\vu}_{\vsigma_{\ma}(\vz)}
\]
and this proves \eqref{eq:J_norm_1}. Combining \eqref{eq:sigma_proof_2}
and \eqref{eq:sigma_proof_3} yields the result.
\end{proof}

\subsection{Algorithm for $\sigma$ Stability}

\label{sub:sigma_stable_algorithm}

Here we prove Theorem~\ref{thm:iteration_solver}. The full pseudocode
for our algorithm, with the exception of how we compute leverage scores
and maintain the inverse of $\ma\mh^{(k)}\ma$ can be seen in Algorithm~\ref{alg:sparseframework}. 

\begin{algorithm2e}[h]
\caption{Algorithm for $\sigma$ Stability Assumption }

\label{alg:sparseframework}

\SetAlgoLined

\textbf{Input: }Initial point $\vec{d}^{(0)}\in\R_{>0}^{n}$.

Set $\vd^{(old)}:=\vec{d}^{(0)}$ and $\gamma\defeq1000c_{s}\log d$
where $c_{s}$ defined in Lemma~\ref{lem:spectral_sparsification}.

Use Lemma~\ref{lem:computing_leverage_scores} to find $\sigma^{(apr)}$
such that $0.99\sigma_{i}^{(apr)}\leq\sigma(\vd^{(0)})_{i}\leq1.01\sigma_{i}^{(apr)}$.

For each $i\in[n]$ : let $h_{i}^{(0)}:=d_{i}/\min\{1,\gamma\cdot\sigma_{i}^{(apr)}\}$
with probability $\min\{1,\gamma\cdot\sigma_{i}^{(apr)}\}$ 

$\quad$and is set to $0$ otherwise.

$\mq^{(0)}\defeq\ma^{T}\mh^{(0)}\ma$.

Let $\mk^{(0)}$ be an approximate inverse of $\mq^{(0)}$ computed
using Theorem~\ref{thm:low_rank_main}.

\textbf{Output: }A $\tilde{O}(d^{2}+\nnz(\ma))$-time linear solver
for $\ma^{T}\md^{(0)}\ma$ (using Theorem~\ref{thm:iteration_solver_strong_promise}
on $\mk^{(0)}$).

\For{each round $k\in[r]$}{

\textbf{Input:} Current point\textbf{ }$\vd^{(k)}\in\R_{>0}^{n}$.

Use Lemma \ref{lem:computing_leverage_scores} and the solver from
the previous round to find $\sigma^{(apr)}$ such that 

$\quad$$0.99\sigma_{i}^{(apr)}\leq\sigma(\vd^{(k)})_{i}\leq1.01\sigma_{i}^{(apr)}$.

\For{each coordinate $i\in[n]$}{

\uIf{ either $0.9\sigma_{i}^{(old)}\leq\sigma_{i}^{(apr)}\leq1.1\sigma_{i}^{(old)}$
or $0.9d_{i}^{(old)}\leq d_{i}^{(k)}\leq1.1d_{i}^{(old)}$ is violated}{

$d_{i}^{(old)}:=d_{i}^{(k)}$.

$\sigma_{i}^{(old)}:=\sigma_{i}^{(apr)}$.

$h_{i}^{(k)}:=d_{i}^{(k)}/\min\{1,\gamma\cdot\sigma_{i}^{(apr)}\}$
with probability $\min\{1,\gamma\cdot\sigma_{i}^{(apr)}\}$ 

$\enspace\enspace\enspace$and is set to $0$ otherwise.

}\Else{

$h_{i}^{(k)}:=h_{i}^{(k-1)}.$

}

}

$\mq^{(k)}\defeq\ma^{T}\mh^{(k)}\ma$.

Let $\mk^{(k)}$ be an approximate inverse of $\mq^{(k)}$ computed
using Theorem~\ref{thm:low_rank_main}.

\textbf{Output:} A linear $\tilde{O}(d^{2}+\nnz(\ma))$-time solver
for $\ma^{T}\md^{(k)}\ma$ (using Theorem~\ref{thm:iteration_solver_strong_promise}
on $\mk^{(k)}$).

}

\end{algorithm2e}

First, we bound the number of coordinates $\mh$ that will change
during the algorithm.
\begin{lem}
\label{lem:alg_diag_change_bound} Suppose that the changes of $\vd$
and the error occurred in computing $\vsigma$ is independent of our
sampled matrix. Under the $\sigma$ stability guarantee, during first
$r$ iterations of Algorithm~\ref{alg:sparseframework}, the expected
number of coordinates changes in $\mh^{(k)}$ over all $k$ is $O(r^{2}\log d)$.\end{lem}
\begin{proof}
Since our error in computing $\sigma$ is smaller than the re-sampling
threshold on how much change of $\sigma$, the re-sampling process
for the $i^{th}$ row happens only when either $\sigma_{\ma}(\vd)_{i}$
or $d_{i}$ changes by more than a multiplicative constant. In order
for re-sampling to affect $\ma^{T}\mh\ma$, it must be the case that
it is either currently in $\ma^{T}\mh\ma$ or about to be put in $\ma^{T}\mh\ma$.
However, since whenever re-sampling occurs the resampling probability
has changed by at most a multiplicative constant, we have that both
these events happen with probability $O(\gamma\cdot\sigma_{\ma}(\vd)_{i})$
using the independence between $\vd$ and the approximate $\vsigma$.
By union bound we have that the probability of sampling row $i$ changing
the matrix $\ma^{T}\mh\ma$ is $O(\sigma_{\ma}(\vd)_{i}\log(d))$.

Observe that whenever we re-sampled the $i^{th}$ row, either $\sigma_{\ma}(\vd)_{i}$
or $d_{i}$ has changed by more than a multiplicative constant. Let
$k_{1}$ be the last iteration we re-sampled the $i^{th}$ row and
let $k_{2}$ be the current iteration. Then, we know that $\sum_{k=k_{1}}^{k_{2}-1}\left|\ln\sigma_{\ma}(\vd^{(k+1)})_{i}-\ln\sigma_{\ma}(\vd^{(k)})_{i}\right|=\Omega(1)$
or $\sum_{k=k_{1}}^{k_{2}-1}\left|\ln\vd_{i}^{(k+1)}-\ln\vd_{i}^{(k)}\right|\geq\Omega(1)$.
Since there are only $r$ steps, we have $\left|k_{2}-k_{1}\right|\leq r$
and hence
\[
\sum_{k=k_{1}}^{k_{2}-1}\left(\ln\sigma_{\ma}(\vd^{(k+1)})_{i}-\ln\sigma_{\ma}(\vd^{(k)})_{i}\right)^{2}=\Omega\left(\frac{1}{r}\right)\text{ or }\sum_{k=k_{1}}^{k_{2}-1}\left(\ln\vd_{i}^{(k+1)}-\ln\vd_{i}^{(k)}\right)^{2}=\Omega\left(\frac{1}{r}\right).
\]
Since $\sigma_{\ma}(\vd)_{i}$ does not change more by a constant
between re-sample (by $\sigma$-stability assumption), we have
\begin{eqnarray*}
\sum_{k=k_{1}}^{k_{2}-1}\sigma_{\ma}(\vd^{(k)})_{i}\left(\ln\sigma_{\ma}(\vd^{(k+1)})_{i}-\ln\sigma_{\ma}(\vd^{(k)})_{i}\right)^{2} & = & \Omega\left(\frac{\sigma_{\ma}(\vd^{(k_{2})})_{i}}{r}\right)\text{ or }\\
\sum_{k=k_{1}}^{k_{2}-1}\sigma_{\ma}(\vd^{(k)})_{i}\left(\ln\vd_{i}^{(k+1)}-\ln\vd_{i}^{(k)}\right)^{2} & = & \Omega\left(\frac{\sigma_{\ma}(\vd^{(k_{2})})_{i}}{r}\right).
\end{eqnarray*}

In summary, re-sampling the $i^{th}$ row indicates that the sum of
the $\sigma$ norm square of the changes of either $\ln\sigma_{\ma}$
or $\ln d$ is at least $\sigma_{\ma}(\vd)_{i}/r$ and with $O(\sigma_{\ma}(\vd)_{i}\log d)$
probability, the sampled matrix is changed. Since $\sum_{k}\norm{\ln\vd^{(k+1)}-\ln\vd^{(k)}}_{\sigma_{\ma}(\vd^{(k)})}^{2}\leq O(r)$
and by Lemma~\ref{lem:lever_score_continuity} $\sum_{k}\norm{\ln\sigma_{\ma}(\vd^{(k+1)})-\ln\sigma_{\ma}(\vd^{(k)})}_{\sigma_{\ma}(\vd^{(k)})}^{2}\leq O(r)$
we have the desired result.
\end{proof}
We now everything we need to prove our main theorem assuming that
the changes of $\vd$ and the error occurred in computing $\sigma$
is independent of our sampled matrix. (In Section~\ref{sec:randomness}
we show how to drop this assumption.)
\begin{proof}[Proof of Theorem \ref{thm:iteration_solver}]
 Note that by design in each iteration $k\in[r]$ we have that $\md^{(old)}\approx_{0.2}\md^{(k)}$,
$\mSigma^{(old)}\approx_{0.2}\mSigma(\vd^{(k)})$ where $\mSigma\defeq\mDiag(\sigma)$.
Thus, we see that the sample probability was chosen precisely so that
we can apply Lemma~\ref{lem:spectral_sparsification}. Hence, we
have $\ma^{T}\mh^{(k)}\ma\approx_{0.1}\ma^{T}\md^{(k)}\ma$. Thus,
the algorithm is correct, it simply remains to bound the running time.

To maintain inverse of $\ma^{T}\mh^{(k)}\ma$, we can simply apply
Theorem \ref{thm:low_rank_main}. By Lemma~\ref{lem:spectral_sparsification},
we know that with high probability $\nnz(\mh^{(0)})\leq\otilde(d)$.
By Lemma~\ref{lem:alg_diag_change_bound} we know there are only
$\tilde{O}(r^{2})$ coordinate changes during the algorithm in expectation.
Therefore, we can use Theorem~\ref{thm:low_rank_main} on $\ma^{T}\mh^{(k)}\ma$
with $\alpha=\tilde{O}(r^{2})$ and $s=\tilde{O}(d)$. Hence, the
average cost of maintain a linear $\tilde{O}(\nnz(\ma)+d^{2})$-time
solver of $\left(\ma^{T}\mh^{(k)}\ma\right)^{-1}$ is $\tilde{O}\left(\frac{d^{\omega}}{r}+dr^{\omega}+r^{2\omega}\right).$
Similar to Theorem \ref{thm:iteration_solver_strong_promise} , we
restart the algorithm every $r=(nd^{\omega-1})^{\frac{1}{2\omega+1}}$
and see that the total cost of maintenance is $\tilde{O}(d^{\omega}+rd^{\frac{2\omega^{2}}{2\omega+1}}).$
Since $\omega\leq1+\sqrt{2}$, we have the total maintenance cost
is $\tilde{O}\left(d^{\omega}+rd^{2}\right).$ 

Thus, all the remains is to bound the cost of computing $\vsigma^{(k)}$.
However, since by the $\sigma$stability assumption $\norm{\log(\vd_{k+1})-\log(\vd_{k})}_{\infty}\leq0.1$
we know that $\md^{(k)}\approx_{0.1}\md^{(k+1)}$ and thus $\mq^{(k)}\approx_{0.2}\ma^{T}\md^{(k+1)}\ma$.
Thus, using Lemma~\ref{lem:computing_leverage_scores} we can compute
$\vsigma^{(k)}$ for all $k\geq2$ in $\otilde(\nnz(\ma)+d^{2}\log\beta)$
time within $0.01$ multiplicative factor. Furthermore, using this
same Lemma and fast matrix multiplication, we can compute $\vsigma^{(1)}$
in $\otilde(\nnz(\ma)+d^{\omega})$ time; therefore, we have our result.

Note that Lemma \ref{lem:alg_diag_change_bound} assume that the changes
of $\vd$ and the error occurred in computing $\sigma$ is independent
of our sampled matrix. Given any linear solver, Theorem \ref{thm:randomness}
in Section~\ref{sec:randomness} shows how to construct a solver
that has same running time up to $\tilde{O}(1)$ factor and is statistically
indistinguishable from the true solver $\left(\ma^{T}\md^{(k)}\ma\right)^{-1}$
and thus circumvent this issue; thereby completing the proof.\end{proof}

%% file: sec_iteration_randomness.tex
\section{Hiding Randomness in Linear System Solvers\label{sec:randomness}}

In many applications (see Section~\ref{sec:applications}), the input
to a particular round of the inverse maintenance problem depends on
our output in the previous round. However, our solution to the inverse
maintenance problem is randomized and if the input to the inverse
maintenance problem was chosen adversarially based on this randomness,
this could break the analysis of our algorithm. Moreover, even within
our solution to the inverse maintenance problem we needed to solve
linear systems and if the output of these linear systems was adversarially
correlated with our randomized computation our analysis would break
(see Section~\ref{sec:sigma_solution})

In this section, we show how to fix both these problems and hide the
randomness we use to approximately solve linear system. We provide
a general transformation, $\mathtt{NoisySolver}$ in Algorithm~\ref{alg:sparseframework_noise},
which turns a linear solver for $\ma^{T}\ma$ into a nonlinear solver
for $\ma^{T}\ma$ that with high probability is indistinguishable
from an exact solver for $\ma^{T}\ma$ plus a suitable Gaussian noise.
The algorithm simply solves the desired linear system using the input
solver and then add a suitable Gaussian noise.

\begin{algorithm2e}
\caption{$\mathtt{NoisySolver}(\vb,\epsilon)$}

\label{alg:sparseframework_noise}

\SetAlgoLined

\textbf{Input: }a linear $\mathcal{T}$-time solver $\solver$ of
$\ma^{T}\ma$, vector $\vb\in\R^{n}$, and accuracy $\epsilon\in(0,1/2)$.

$\vy_{1}:=\solver(\vb,\epsilon_{1})$ where $\epsilon_{1}=\varepsilon\left(32n^{7}\right)^{-2}$.

Let $\veta\in\mathbb{R}^{n}$ be sampled randomly from the normal
distribution with mean $\vzero$ and covariance $\iMatrix$.

$\vy_{2}:=\solver(\ma^{T}\veta,\epsilon_{2})$ where $\epsilon_{2}=\left(12dn^{6}\right)^{-2}$.

\textbf{Output:} $\vy\defeq\vy_{1}+\alpha\vy_{2}$ where $\alpha=\frac{1}{8}\sqrt{\frac{\varepsilon}{n}}\normFull{\vy_{1}}_{\ma^{T}\ma}.$

\end{algorithm2e}

We break the proof that this works into two parts. First, in Lemma~\ref{lem:noisey_solver_correctness}
we show that $\mathtt{NoisySolver}$, is in fact a solver and then
in Theorem~\ref{thm:randomness} we show that with high probability
it is indistinguishable from an exact solver plus a Gaussian noise. 
\begin{lem}
\label{lem:noisey_solver_correctness} Let $\ma\in\R^{n\times d}$
and let $\solver$ be a linear $\time$-time solver for $\ma^{T}\ma$.
For all $\vb\in\R^{n}$ and $\epsilon\in(0,1/2)$, the algorithm $\mathtt{NoisySolver}(\vb,\varepsilon)$
is a $(\nnz(\ma)+\time)$-time solver for $\ma^{T}\ma$.\end{lem}
\begin{proof}
By the definition of $\vy$ and the inequality $(a+b)^{2}\leq2a^{2}+2b^{2}$,
we have 
\begin{eqnarray*}
\normFull{\vy-\left(\ma^{T}\ma\right)^{-1}\vb}_{\ma^{T}\ma}^{2} & \leq & 2\norm{\vy_{1}-\left(\ma^{T}\ma\right)^{-1}\vb}_{\ma^{T}\ma}^{2}+2\alpha^{2}\norm{\vy_{2}}_{\ma^{T}\ma}^{2}.
\end{eqnarray*}
To bound the first term, recall that by the definition of $\solver$
\[
\norm{\vy_{1}-\left(\ma^{T}\ma\right)^{-1}\vb}_{\ma^{T}\ma}^{2}\leq\epsilon_{1}\norm{\left(\ma^{T}\ma\right)^{-1}\vb}_{\ma^{T}\ma}^{2}.
\]
To bound the second term, we note that by the definition of $\veta$,
$\norm{\veta}_{2}^{2}$ follows $\chi^{2}$-distribution with $n$
degrees of freedom. It is known that \cite[Lem 1]{laurent2000adaptive}
for all $t>0$, 
\[
\mathbb{P}\left(\norm{\veta}_{2}^{2}>n+2\sqrt{nt}+2t\right)\leq\exp(-t)\enspace.
\]
Hence, with high probability in $n$, $\norm{\veta}_{2}^{2}<2n.$
Using this and the definition of $\solver$ yields
\begin{eqnarray*}
\norm{\vy_{2}}_{\ma^{T}\ma}^{2} & \leq & 2\norm{\vy_{2}-\left(\ma^{T}\ma\right)^{-1}\ma^{T}\veta}_{\ma^{T}\ma}^{2}+2\norm{\left(\ma^{T}\ma\right)^{-1}\ma^{T}\veta}_{\ma^{T}\ma}^{2}\\
 & \leq & 2\left(1+\epsilon_{2}\right)\norm{\left(\ma^{T}\ma\right)^{-1}\ma^{T}\veta}_{\ma^{T}\ma}^{2}\\
 & \leq & 4\norm{\veta}_{2}^{2}\leq8n.
\end{eqnarray*}
where we used that $\epsilon_{2}\leq1$ and $\norm{\ma\left(\ma^{T}\ma\right)^{-1}\ma^{T}}_{2}\leq1$.
Using that $\epsilon_{1}\leq1$ and applying a similar proof yields
that
\[
\alpha^{2}=\frac{\varepsilon}{64n}\norm{\vy_{1}}_{\ma^{T}\ma}^{2}\leq\frac{\epsilon}{32n}\normFull{\left(\ma^{T}\ma\right)^{-1}\vb}_{\ma^{T}\ma}^{2}
\]
Consequently, with high probability in $n$,
\[
\normFull{\vy-\left(\ma^{T}\ma\right)^{-1}\vb}_{\ma^{T}\ma}^{2}\leq2\epsilon_{1}\norm{\left(\ma^{T}\ma\right)^{-1}\vb}_{\ma^{T}\ma}^{2}+16\alpha^{2}n\leq\epsilon\norm{\left(\ma^{T}\ma\right)^{-1}\vb}_{\ma^{T}\ma}^{2}\enspace.
\]
\end{proof}
\begin{thm}
\label{thm:randomness}Let $\mathtt{IdealSolver}(\vb)=\left(\ma^{T}\ma\right)^{-1}\vb+\vc$
where $\vc$ follows the normal distribution with mean $0$ and covariance
$\beta^{2}\left(\ma^{T}\ma\right)^{-1}$ where $\beta=\frac{1}{8}\sqrt{\frac{\varepsilon}{n}}\normFull{\left(\ma^{T}\ma\right)^{-1}\vb}_{\ma^{T}\ma}.$
Then, the total variation distance between the outcome of $\mathtt{NoisySolver}$
and $\mathtt{IdealSolver}$ is less than $1/n^{3}.$ Therefore, any
algorithm calls $\mathtt{IdealSolver}$ less than $n^{2}$ times cannot
distinguish between $\mathtt{NoisySolver}$ and $\mathtt{IdealSolver}$.\end{thm}
\begin{proof}
Since $\solver$ is linear, we have $\vy_{2}=\mq_{\varepsilon_{2}}\ma^{T}\veta$
for some matrix $\mq_{\varepsilon_{2}}$. Since $\veta$ follows the
normal distribution with mean $0$ and covariance matrix $\iMatrix$,
we have
\[
\mathbb{E}\left[\vy_{2}\vy_{2}^{T}\right]=\mq_{\varepsilon_{2}}\ma^{T}\mathbb{E}\left[\veta\veta^{T}\right]\ma\mq_{\varepsilon_{2}}^{T}=\mq_{\varepsilon_{2}}\ma^{T}\ma\mq_{\varepsilon_{2}}^{T}
\]
Lemma \ref{lem:solver_implies_approx} shows that $\mq_{\varepsilon_{2}}^{T}\ma^{T}\ma\mq_{\varepsilon_{2}}\approx_{4\sqrt{\varepsilon_{2}}}(\ma^{T}\ma)^{-1}$
and hence
\begin{equation}
\mathbb{E}\left[\vy_{2}\vy_{2}^{T}\right]\approx_{4\sqrt{\varepsilon_{2}}}\left(\ma^{T}\ma\right)^{-1}.\label{eq:QAAQ_error}
\end{equation}
Since $\vy_{2}$ is a linear transformation of $\veta$, $\vy_{2}$
follows is given by the normal distribution with mean $0$ and covariance
$\mq_{\varepsilon_{2}}\ma^{T}\ma\mq_{\varepsilon_{2}}^{T}$, i.e.
$\vy_{2}\in\mathcal{N}(0,\mq_{\varepsilon_{2}}\ma^{T}\ma\mq_{\varepsilon_{2}}^{T})$.

Now let, $\vy,\vz$ be the output of $\mathtt{NoisySolver}(\vb)$
and $\mathtt{IdealSolver}(\vb)$ respectively; we know that 
\begin{eqnarray*}
\vy & \in & \mathcal{N}\left(\vy_{1},\alpha^{2}\mq_{\varepsilon_{2}}\ma^{T}\ma\mq_{\varepsilon_{2}}^{T}\right)\text{ and }\\
\vz & \in & \mathcal{N}\left(\left(\ma^{T}\ma\right)^{-1}\vb,\beta^{2}\left(\ma^{T}\ma\right)^{-1}\right)
\end{eqnarray*}
where $\alpha=\frac{1}{8}\sqrt{\frac{\varepsilon}{n}}\normFull{\vy_{1}}_{\ma^{T}\ma}$
and $\beta=\frac{1}{8}\sqrt{\frac{\varepsilon}{n}}\normFull{\left(\ma^{T}\ma\right)^{-1}\vb}_{\ma^{T}\ma}$.
To bound $\norm{\vy-\vz}_{\mathrm{tv}}$, Pinsker's inequality shows
that $\norm{\vy-\vz}_{\mathrm{tv}}\leq\sqrt{\frac{1}{2}D_{\mathrm{KL}}(\vy||\vz)}$
and using an explicit formula for the KL divergence for the normal
distribution we in turn have that $\norm{\vy-\vz}_{\mathrm{tv}}$
is bounded by
\[
\frac{1}{2}\sqrt{\mathrm{tr}\left(\left(\frac{\alpha}{\beta}\right)^{2}\ma^{T}\ma\mq_{\varepsilon_{2}}\ma^{T}\ma\mq_{\varepsilon_{2}}^{T}\right)+\norm{\vy_{1}-\left(\ma^{T}\ma\right)^{-1}\vb}_{\beta^{-2}\ma^{T}\ma}^{2}-d+\ln\left(\frac{\det\beta^{2}\left(\ma^{T}\ma\right)^{-1}}{\det\alpha^{2}\mq_{\varepsilon_{2}}\ma^{T}\ma\mq_{\varepsilon_{2}}^{T}}\right)}
\]
Hence, we need to prove $\alpha^{2}\mq_{\varepsilon_{2}}\ma^{T}\ma\mq_{\varepsilon_{2}}^{T}\approx\beta^{2}\left(\ma^{T}\ma\right)^{-1}$
and $\vy_{1}\approx\left(\ma^{T}\ma\right)^{-1}\vb$

First, we first prove $\vy_{1}\approx\left(\ma^{T}\ma\right)^{-1}\vb$:
\begin{align*}
\norm{\vy_{1}-\left(\ma^{T}\ma\right)^{-1}\vb}_{\beta^{-2}(\ma^{T}\ma)}^{2} & \leq\beta^{-2}\varepsilon_{1}\normFull{\left(\ma^{T}\ma\right)^{-1}\vb}_{\ma^{T}\ma}^{2}\tag{Definition of \ensuremath{\solver}and \ensuremath{\vy_{1}}}\\
 & \leq\frac{64n\varepsilon_{1}}{\varepsilon}\,.\tag{Definition of \ensuremath{\beta}}
\end{align*}
Next, to prove $\alpha^{2}\mq_{\varepsilon_{2}}\ma^{T}\ma\mq_{\varepsilon_{2}}^{T}\approx\beta^{2}\left(\ma^{T}\ma\right)^{-1}$,
we note that by triangle inequality and the definition of $\solver$
and $\vy_{1}$ we have
\[
\left(1-\sqrt{\varepsilon_{1}}\right)\normFull{\left(\ma^{T}\ma\right)^{-1}\vb}_{\ma^{T}\ma}\leq\normFull{\vy_{1}}_{\ma^{T}\ma}\leq\left(1+\sqrt{\varepsilon_{1}}\right)\normFull{\left(\ma^{T}\ma\right)^{-1}\vb}_{\ma^{T}\ma}.
\]
Therefore by the definition of $\alpha$ and $\beta$ we have 
\[
\left(1-3\sqrt{\varepsilon_{1}}\right)\leq\frac{\alpha^{2}}{\beta^{2}}\leq\left(1+3\sqrt{\varepsilon_{1}}\right).
\]
Using \eqref{eq:QAAQ_error} then yields that
\[
\alpha^{2}\mq_{\varepsilon_{2}}\ma^{T}\ma\mq_{\varepsilon_{2}}^{T}\approx_{4\sqrt{\varepsilon_{2}}+4\sqrt{\varepsilon_{1}}}\beta^{2}\left(\ma^{T}\ma\right)^{-1}.
\]
Therefore, we have 
\begin{eqnarray*}
\mathrm{tr}\left((\frac{\alpha}{\beta})^{2}\ma^{T}\ma\mq_{\varepsilon_{2}}\ma^{T}\ma\mq_{\varepsilon_{2}}^{T}\right) & = & \mathrm{tr}\left((\frac{\alpha}{\beta})^{2}\left(\ma^{T}\ma\right)^{1/2}\mq_{\varepsilon_{2}}\ma^{T}\ma\mq_{\varepsilon_{2}}^{T}\left(\ma^{T}\ma\right)^{1/2}\right)\\
 & \leq & de^{4\sqrt{\varepsilon_{2}}+4\sqrt{\varepsilon_{1}}}\\
 & \leq & d+8\sqrt{\varepsilon_{2}}d+8\sqrt{\varepsilon_{1}}d
\end{eqnarray*}
and
\begin{eqnarray*}
\ln\left(\frac{\det\beta^{2}\left(\ma^{T}\ma\right)^{-1}}{\det\alpha^{2}\mq_{\varepsilon_{2}}\ma^{T}\ma\mq_{\varepsilon_{2}}^{T}}\right) & = & \ln\det\left((\frac{\alpha}{\beta})^{2}\left(\ma^{T}\ma\right)^{1/2}\mq_{\varepsilon_{2}}\ma^{T}\ma\mq_{\varepsilon_{2}}^{T}\left(\ma^{T}\ma\right)^{1/2}\right)^{-1}\\
 & \leq & d\ln\left(e^{4\sqrt{\varepsilon_{2}}+4\sqrt{\varepsilon_{1}}}\right)\\
 & = & 4\sqrt{\varepsilon_{2}}d+4\sqrt{\varepsilon_{1}}d.
\end{eqnarray*}

Combining inequalities above and using our choice of $\varepsilon_{1}$
and $\varepsilon_{2}$ we have
\[
\norm{\vy-\vz}_{\mathrm{tv}}\leq\frac{1}{2}\sqrt{12\sqrt{\varepsilon_{2}}d+12\sqrt{\varepsilon_{1}}d+64n\frac{\varepsilon_{1}}{\varepsilon}}\leq1/n^{3}\,.
\]
\end{proof}

%% file: sec_iteration_insertion_removal.tex
\section{Row Insertion and Removal\label{sec:iteration_insertion_removal}}

For some applications we need to solve the inverse maintenance problem
when the matrix $\ma$ might change between rounds. In particular
some rows of $\ma$ might be entirely removed or some new rows might
be added. 

For example, in each iteration of cutting plane methods, a new constraint
is added to the current polytope $\{\ma\vx\geq\vb\}$; each iteration
of quasi newton methods, the current approximate Hessian is updated
by a rank 1 matrix. If the matrix is updated only by a rank 1 matrix
each iteration, the inverse can be updated efficiently using the Sherman-Morrison
formula. However, for the general case, it is less obvious when the
matrix $\ma$ and the diagonal $\md$ can be changed by a high rank
matrix. 

Here we formally define the more general set of assumptions under
which we would like to solve the inverse maintenance problem and show
how to solve this problem efficiently. 
\begin{defn}[$K$ Stability Assumption]
\label{ass:stability_sigma_2} The inverse maintenance problem satisfies
the $K$ stability assumption if 
\begin{enumerate}
\item A row of $\ma$ is revealed to the algorithm at round $k$ if and
only if the corresponding entry in $\vd^{(k)}$ is non-zero. 
\item For each $k\in[r]$, we have either

\begin{enumerate}
\item $\norm{\log(\vd^{(k)})-\log(\vd^{(k-1)})}_{\vLever_{\ma}(\vd^{(k)})}\leq0.1$
and $\norm{\log(\vd^{(k)})-\log(\vd^{(k-1)})}_{\infty}\leq0.1$; or
\item The vectors $\vd^{(k)}$ and $\vd^{(k-1)}$ differ in at most $K$
coordinates.
\end{enumerate}
\item $\beta^{-1}\ma^{T}\md^{(0)}\ma\preceq\ma^{T}\md^{(k)}\ma\preceq\beta\ma^{T}\md^{(0)}\ma$
for $\beta=\poly(n)$. 
\end{enumerate}
\end{defn}
Some of the previous algorithms for the inverse maintenance, such
as \cite{vaidya1989speeding}, rely on the assumption that the entire
matrix $\ma$ is given explicitly. In particular these algorithms
perform precomputation on the entirety of $\ma$ that they use to
decrease the amortize cost of later steps. Here we show that the Algorithm~\ref{alg:sparseframework}
we proposed does not have this drawback and can be easily modified
to solve this version of inverse maintenance problem under the $K$
stability assumption for a fairly large $K$.
\begin{thm}
\label{thm:iteration_solver_2} Suppose that the inverse maintenance
problem satisfies the $K$ stability assumption for $K\leq d^{(3-\omega)/(\omega-1)}$,
where $\omega$ is the matrix multiplication constant. Then there
is a variant of Algorithm~\ref{alg:sparseframework} that maintains
a $\tilde{O}(\nnz(\ma)+d^{2})$-time solver with high probability
in total time $\otilde(d^{\omega}+r(\nnz(\ma)+d^{2}))$ where $r$
is the number of rounds.\end{thm}
\begin{proof}
From our solution to the problem under the $\ellTwo$ stability assumption,
i.e. the proof of Theorem~\ref{thm:low_rank_main}, we see that we
do not need to know the entire matrix $\ma$ to maintain an approximate
inverse provided that the number of low rank updates we need to perform
is small. Therefore, it suffices to show how to reduce the $\sigma$
stability assumption with row addition and removal to the low rank
update problem mentioned in Theorem~\ref{thm:low_rank_main}. 

Our previous reduction under the $\sigma$ stability assumption, i.e.
Theorem~\ref{thm:iteration_solver}, relied on two facts, (1) that
$\ma^{T}\md^{(k-1)}\ma\approx_{O(1)}\ma^{T}\md^{(k)}\ma$ for all
$k$ and that (2) the expected number of row changes is $O(r^{2})$.
The first condition, (1) was used to ensure that we could efficiently
compute approximate leverage scores for $\ma^{T}\md^{(k)}\ma$ and
the second condition (2) was used to ensure our invocation of Theorem~\ref{thm:low_rank_main}
was efficient. In the remainder of this proof we show how both condition
(1) and (2) can be achieved.

First we show how to achieve condition (1) under changes to $\vd^{(k)}.$
First note that if the change follows case (a) of the $K$ stability
of assumption then since $\norm{\log(\vd^{(k)})-\log(\vd^{(k-1)})}_{\infty}\leq0.1$
clearly $\ma^{T}\md^{(k-1)}\ma\approx_{O(1)}\ma^{T}\md^{(k)}\ma$
as before. On the other hand, for case (b) we claim that (1) can be
achieved by adding intermediate steps into the problem as follows.
First, we can assume $\vd^{(k)}\geq\vd^{(k-1)}$ or $\vd^{(k)}\leq\vd^{(k-1)}$
by splitting changes that corresponds to inserting or removing rows
into two steps. If we do the insertion first, it is easy to prove
that $\gamma^{-1}\ma^{T}\md^{(0)}\ma\preceq\ma^{T}\md^{(k)}\ma\preceq\gamma\ma^{T}\md^{(0)}\ma$
still holds for some $\gamma=O(\beta)$ after the splitting. For the
insertion case, we consider the matrix
\[
\mh(\alpha)=\ma^{T}\left(\md^{(k-1)}+\alpha\left(\md^{(k)}-\md^{(k-1)}\right)\right)\ma.
\]
Since $\md^{(k)}-\md^{(k-1)}\succeq\mZero$, it is easy to see that
$\mh(\alpha)\approx_{O(1)}\mh(2\alpha)$. Since $\mh(1)\approx_{\poly(n)}\mh(0)$,
as we have just argued, we have $\ma^{T}\left(\md^{(k)}-\md^{(k-1)}\right)\ma\approx_{\poly(n)}\mh(0)$
and hence $\mh(\poly(1/n))\approx_{O(1)}\mh(0)$. Therefore, one can
split the insertion case into $O(\log(n))$ steps to ensure $\ma^{T}\md^{(k-1)}\ma\approx_{O(1)}\ma^{T}\md^{(k)}\ma$.
The argument for removal is analogous.

Next we show how to achieve condition (2). Again we consider the case
where a round is updated by (a) and by (b) separately. For (a), the
changes of $\vd$ is small in $\sigma$ norm, we already know from
our previous analysis of the $\sigma$ stability assumption that it
can only cause $O(r^{2})$ many rows change. For case (b) we define
$\alpha$ be the total number of row changes due to (b). From the
definition of Algorithm~\ref{alg:sparseframework}, we see that there
are two causes of resampling, namely, the change of $\vd$ and the
change of $\vsigma$. By the assumption (b) we know that there are
at most $K$ rows changes in $\vd$ and therefore, this contributes
at most $rK$ to $\alpha$. For the change of $\vsigma$, from first
half of the proof of Lemma~\ref{lem:alg_diag_change_bound}, we know
that the number of resampling that causes changes in $\vh$ is bounded
by 

\[
\sum_{i}\max\left\{ \sigma_{i}(\vd^{(k)}),\sigma_{i}(\vd^{(k+1)})\right\} \min\left\{ \left|\ln\sigma_{i}(\vd^{(k)})-\ln\sigma_{i}(\vd^{(k+1)})\right|,1\right\} .
\]
Therefore, we have that
\begin{equation}
\alpha\leq rK+\sum_{i,k}\max\left\{ \sigma_{i}(\vd^{(k)}),\sigma_{i}(\vd^{(k+1)})\right\} \min\left\{ \left|\ln\sigma_{i}(\vd^{(k)})-\ln\sigma_{i}(\vd^{(k+1)})\right|,1\right\} .\label{eq:sec_insertion_alpha_bound}
\end{equation}
For any $k$, we define 
\[
S_{k}=\left\{ i\text{ such that }\frac{1}{2}\leq\sigma_{i}(\vd^{(k+1)})/\sigma_{i}(\vd^{(k)})\leq2\right\} .
\]
Then, we have that
\begin{eqnarray*}
\alpha & \leq & rK+O\left(1\right)\sum_{k}\sum_{i\in S_{k}}\left|\sigma_{i}(\vd^{(k)})-\sigma_{i}(\vd^{(k+1)})\right|+O\left(1\right)\sum_{k}\sum_{i\notin S_{k}}\max\left\{ \sigma_{i}(\vd^{(k)}),\sigma_{i}(\vd^{(k+1)})\right\} \\
 & \leq & rK+O\left(1\right)\sum_{k}\sum_{i\in S_{k}}\left|\sigma_{i}(\vd^{(k)})-\sigma_{i}(\vd^{(k+1)})\right|+O\left(1\right)\sum_{k}\sum_{i\notin S_{k}}\left|\sigma_{i}(\vd^{(k+1)})-\sigma_{i}(\vd^{(k)})\right|\\
 & \leq & rK+O\left(1\right)\sum_{i,k}\left|\sigma_{i}(\vd^{(k)})-\sigma_{i}(\vd^{(k+1)})\right|.
\end{eqnarray*}
Again, we can split the second case into insertion only and removal
only. For the insertion case, we note that $\sigma_{i}(\vd^{(k+1)})\leq\sigma_{i}(\vd^{(k)})$
for all $i$ with $\vd_{i}^{(k+1)}=\vd_{i}^{(k)}$. Since $\sum\sigma_{i}(\vd^{(k+1)})=\sum\sigma_{i}(\vd^{(k)})=d$,
we have that
\begin{eqnarray*}
\sum_{i}\left|\sigma_{i}(\vd^{(k)})-\sigma_{i}(\vd^{(k+1)})\right| & = & 2\sum_{\vd_{i}^{(k+1)}\neq\vd_{i}^{(k)}}\left|\sigma_{i}(\vd^{(k)})-\sigma_{i}(\vd^{(k+1)})\right|\\
 & \leq & 4\left|\left\{ i\text{ such that }\vd_{i}^{(k+1)}\neq\vd_{i}^{(k)}\right\} \right|.
\end{eqnarray*}
We have the same bound for the removal case. Putting it into \eqref{eq:sec_insertion_alpha_bound},
we have
\[
\alpha\leq rK+O\left(1\right)\sum_{k}\left|\{i\text{ such that }\vd_{i}^{(k+1)}\neq\vd_{i}^{(k)}\}\right|\leq O(rK)\,.
\]
This proves that the total expected number of row changes is bounded
by $O(r^{2}+rK)$. Note that we use the fact that if there is an update
triggered by case (a) after the update is triggered by case (b) then
we account for it by case (b), i.e. the interactions between case
(a) and (b) only increase the number of row changes by a constant.

Now, we can apply Theorem \ref{thm:low_rank_main} and restart every
$r=d^{\omega-2}$ steps and get that the average cost of maintenance
is
\begin{align*}
 & \tilde{O}\left(d^{2}+(r^{2}+rK)d\left(\frac{r^{2}+rK}{r}\right)^{\omega-2}+(r^{2}+rK)^{\omega}\right)\\
 & =\otilde\left(d^{2}+dr^{2(\omega-1)-(\omega-2)}+r^{\omega-1-(\omega-2)}K^{\omega-1}+r^{2\omega}+r^{\omega}K^{\omega}\right)\\
 & =\otilde\left(d^{2}+d^{1+(\omega-2)\omega}+d^{\omega-2}K^{\omega-1}+d^{2\omega(\omega-2)}+d^{\omega(\omega-2)}K^{\omega}\right)
\end{align*}
Now using that $\omega\leq1+\sqrt{2}$ we have that $1+(\omega-2)\omega\leq2$
and $2\omega(\omega-2)\leq2$ and therefore the average cost of maintenance
is 
\[
\otilde\left(d^{2}+d^{\omega-2}K^{\omega-1}+d^{\omega(\omega-2)}K^{\omega}\right)
\]
Hence, as long as $K\leq d^{(3-\omega)/(\omega-1)}$, we have that
the average cost is $\tilde{O}(d^{2})$.\end{proof}
\begin{rem}
Using $\omega<2.37287$ \cite{gall2014powers}, the above theorem
shows how to solve inverse maintenance problem with $d^{0.4568}$
rows addition and removal in amortized $\tilde{O}(d^{2})$ time.\end{rem}

%% file: sec_iteration_applications.tex
\section{Applications}

\label{sec:applications}

In this section, we provide multiple applications of our algorithm
for solving the inverse maintenance problem under the $\ellTwo$ and
$\sigma$ stability assumptions. First in Section~\ref{sub:linear_programming}
we show how to use our results to solve a linear program. Then in
Section~\ref{sub:app:regression}, Section~\ref{sub:app:multicommodity_flow},
and Section~\ref{sub:app:mincost_flow} we state the consequences
of this linear programming algorithm for solving the non-linear regression,
multicommodity flow, and minimum cost flow respectively. In Section~\ref{sub:app:rounding}
we show how our results lead to new efficient algorithms for computing
a rounding of a polytope.

\subsection{Linear Programming\label{sub:linear_programming}}

Here we show how to use our solution to the inverse maintenance problem
under the $\sigma$ stability assumption (See Section~\ref{sec:sigma_solution})
to improve the running time of solving a linear program. Our main
result is the following:
\begin{thm}
\label{thm:LPSolve} Let $\ma\in\R^{d\times n}$, $\vc,\vl,\vu\in\R^{n}$,
and $\vb\in\R^{d}$ for $d\leq n$ and suppose $\vx\in\R^{n}$ is
an interior point of the polytope
\[
S=\left\{ \vx\in\R^{n}\,:\,\ma\vx=\vb\text{ and }l_{i}\leq x_{i}\leq u_{i}\text{ for all }i\in[n]\right\} \enspace.
\]
Let $U\defeq\max\left(\normFull{\frac{\vu-\vl}{\vu-\vx}}_{\infty},\normFull{\frac{\vu-\vl}{\vx-\vl}}_{\infty},\norm{\vu-\vl}_{\infty},\norm{\vc}_{\infty}\right)$
. Then, consider the linear program
\begin{equation}
\text{OPT}\defeq\min_{\vx\in S}\vc^{T}\vx\enspace.\label{eq:LP}
\end{equation}
In time $\tilde{O}\left(\sqrt{d}\left(\nnz(\ma)+d^{2}\right)\log\left(U/\epsilon\right)\right)$,
we can compute $\vy$ such that $\norm{\ma\vy-\vb}_{\ma^{T}\ms^{-2}\ma}\leq\epsilon$,
$l_{i}\leq x_{i}\leq u_{i}$ and $\vc^{T}\vy\leq OPT+\epsilon$ where
$\ms$ is a diagonal matrix $\ms_{ii}=\min\left(u_{i}-y_{i},y_{i}-l_{i}\right)$.\end{thm}
\begin{proof}
In \cite[ArXiv v2, Thm 28]{lsMaxflow}, we showed how to solve linear
program of the form \eqref{eq:LP} by solving a sequence of slowly
changing linear systems $\ma^{T}\md^{(k)}\ma\vx=\vq^{(k)}$ where
$\md^{(k)}$ is a diagonal matrix corresponding to a weighted distance
of $\vx_{k}$ to the boundary of the polytope. In \cite[ArXiv v2, Lem 32]{lsMaxflow}
we showed that this sequence of linear systems satisfied the $\sigma$
stability assumption. Furthermore, we showed that it suffices to solve
these linear systems to $1/\poly(n)$ accuracy. Since, the algorithm
consists of $\sqrt{d}\log(U/\epsilon)$ rounds of this algorithm plus
additional $\tilde{O}(\nnz(\ma)+d^{2})$ time per round we have the
desired result.
\end{proof}
We remark that we can only output an almost feasible point but it
is difficult to avoid because finding any point $\vx$ such that $\ma\vx=\vb$
takes $O(nd^{\omega-1})$ which is slower than our algorithm when
$n\gg d$. Similarly, we have an algorithm for the dual of \eqref{eq:LP}
as follows:
\begin{thm}
\label{thm:LPSolve_dual} Let $\ma\in\R^{n\times d}$ where $d\leq n$.
Suppose we have an initial point $\vx\in\R^{n}$ such that $\ma^{T}\vx=\vb$
and $-1<x_{i}<1$. Then, we can find $\vy\in\R^{d}$ such that 
\begin{equation}
\vb^{T}\vy+\norm{\ma\vy+\vc}_{1}\leq\min_{\vy}\left(\vb^{T}\vy+\norm{\ma\vy+\vc}_{1}\right)+\epsilon.\label{eq:LP_dual}
\end{equation}
in $\tilde{O}\left(\sqrt{d}\left(\nnz(\ma)+d^{2}\right)\log\left(U/\epsilon\right)\right)$
time where $U=\max\left(\normFull{\frac{2}{1-\vx}}_{\infty},\normFull{\frac{2}{\vx+1}}_{\infty},\norm{\vc}_{\infty}\right)$. \end{thm}
\begin{proof}
It is same as Theorem \ref{thm:LPSolve} except we invoke \cite[ArXiv v2, Thm 29]{lsMaxflow}.\end{proof}
\begin{rem}
The existence of the interior point in Theorem \ref{thm:LPSolve_dual}
certifies the linear program has bounded optimum value. Standard tricks
can be used to avoid requiring such interior point but may yield a
more complicated looking running time (see Appendix~E of \cite{lsInteriorPoint}
for instance).
\end{rem}

\subsection{$\ell^{1}$ and $\ell^{\infty}$ Regression \label{sub:app:regression}}

The $\ell^{p}$ regression problem involves finding a vector $\vx$
that minimize $\norm{\ma\vx-\vc}_{p}$ for some $n\times d$ matrix
$\ma$ and some vector $\vc$. Recently, there has been much research
\cite{mahoney2012fast,nelson2012osnap,li2012iterative,clarkson2013low,Cohen2014}
on solving overdetermined problems (i.e. $n\gg d$) as these arises
naturally in applications involving large datasets. While there has
been recent success on achieving algorithms whose running time is
nearly linear input plus something polynomial in $d$, in the case
that $p\neq2$ these algorithms achieve a polynomial dependence on
the desired accuracy $\epsilon$ \cite{cohen2014ell_p}. Here we show
how to improve the dependence on $\epsilon$ by paying a multiplicative
$\sqrt{d}$ factor and improve upon previous algorithms in this regime.
\begin{cor}
Let $\ma\in\R^{n\times d}$, $\vc\in\mathbb{R}^{n}$, and $p=1$ or
$p=\infty$. There is an algorithm to find $\vx$ such that
\[
\norm{\ma\vx-\vc}_{p}\leq\min_{\vy}\norm{\ma\vy-\vc}_{p}+\epsilon\norm{\vc}_{p}
\]
in $\tilde{O}\left(\sqrt{d}\left(\nnz(\ma)+d^{2}\right)\log\left(\epsilon^{-1}\right)\right)$
time.\end{cor}
\begin{proof}
The $\ell^{1}$ case is the special case of Theorem \ref{thm:LPSolve_dual}
with $\vb=\vec{0}$ and an explicit initial point $\vec{0}$. 

For the $\ell^{\infty}$ case, we consider the following linear program
\begin{equation}
\min_{\vx,t}\left(-2n+\frac{1}{2}\right)t+\norm{\ma\vx-\vc-t\onesVec}_{1}+\norm{\ma\vx-\vc+t\onesVec}_{1}.\label{eq:LP_linf}
\end{equation}
where $\onesVec\in\R^{n}$ is the all ones vector. For all $t>0$,
we have 
\[
\left|a-t\right|+\left|a+t\right|=\begin{cases}
2t+(a-t) & \text{if }a\geq t\\
2t & \text{if }-t\leq a\leq t\\
2t+(-t-a) & \text{if }a\leq-t
\end{cases}.
\]
Letting $\mathrm{dist}(a,[-t,t])$ denote the distance from $a$ to
the interval $[-t,t]$ (and $|a|+|t|$ if $t\leq0$) we then have
that $\left|a-t\right|+\left|a+t\right|=dist(a,[-t,t])+2t$ and consequently
the linear program \eqref{eq:LP_linf} is equivalent to
\[
\min_{\vx,t}f(\vx,t)\defeq\frac{t}{2}+\sum_{i=1}^{n}\mathrm{dist}\left(\left[\ma\vx-c\right]_{i},[-t,t]\right)\,.
\]
Note that when $t\geq\norm{\ma\vx-\vc}_{\infty}$ we have $f(\vx,t)=\frac{t}{2}$
and when $t\leq\norm{\ma\vx-\vc}_{\infty}$ we have $f(\vx,t)\geq\frac{1}{2}\norm{\ma\vx-\vc}_{\infty}$.
Consequently, the linear program \eqref{eq:LP_linf} is equivalent
to $\ell_{\infty}$ regression. To solve \eqref{eq:LP_linf}, we rewrite
it as follows
\begin{equation}
\min_{\vx,t}\left(-2n+\frac{1}{2}\right)t+\normFull{\left(\begin{array}{cc}
\ma & -\onesVec\\
\ma & \onesVec
\end{array}\right)\left(\begin{array}{c}
\vx\\
t
\end{array}\right)-\vc}_{1}.\label{eq:LP_linf2}
\end{equation}
Since
\[
\left(\begin{array}{cc}
\ma & -\onesVec\\
\ma & \onesVec
\end{array}\right)^{T}\vy=\left(\begin{array}{c}
\vec{0}_{d}\\
-2n
\end{array}\right)\enspace\text{ for }\enspace\vy\defeq\left(\begin{array}{c}
\onesVec\\
-\onesVec
\end{array}\right),
\]
we can use $\frac{2n-\frac{1}{2}}{2n}\vy$ as the initial point for
Theorem~\ref{thm:LPSolve_dual} and apply it to \eqref{eq:LP_linf}
to find $\vx,t$ as desired.\end{proof}
\begin{rem}
We wonder if it is possible to obtain further running time improvements
for solving $\ell^{p}$ regression when $p\notin\{1,2,\infty\}$.
\end{rem}

\subsection{$\tilde{O}(d)$ Rounding Ellipsoid for Polytopes \label{sub:app:rounding}}

For any convex set $K$, we call an ellipsoid $E$ is an $\alpha$-rounding
if $E\subset K\subset\alpha E$. It is known that every convex set
in $d$ dimension has a $d$-rounding ellipsoid and that such rounding
have many applications. (See \cite{khachiyan1996rounding,vempala2010recent})
For polytopes $K=\{\vx\in\R^{d}\,:\,\ma\vx\geq\vb\}$, the previous
best algorithm for finding a $(1+\epsilon)d$ rounding takes time
$\tilde{O}(n^{3.5}\epsilon^{-1})$ \cite{khachiyan1996rounding} and
$\tilde{O}(nd^{2}\epsilon^{-1})$ \cite{kumar2005minimum}. Here we
show how to compute an $O(d)$ rounding in time $\tilde{O}(\sqrt{d}(\nnz(\ma)+d^{2})\log(U))$,
which is enough for many applications. Note that this is an at least
$O(\sqrt{d})$ improvement over previous results and for the case
$\ma$ is sparse, this is an $O(d^{1.5})$ improvement. 

We split our proof into two parts. First we provide a technical lemma,
Lemma~\ref{lem:John_ellipsoid_gurantee}, which shows conditions
under which we a point in a polytope is the center of a suitable ellipse
and in Theorem~\ref{thm:computing_rounding} we show how to us this
Lemma to compute the ellipse.
\begin{lem}
\label{lem:John_ellipsoid_gurantee} Let $P\defeq\{\vx\in\R^{d}\,:\,\ma\vx\geq\vb\}$
be a polytope for $\ma\in\R^{n\times d}$ and $d\leq n$. Furthermore,
let $\vw\in\R_{>0}^{n}$ and let $p(\vx)\defeq-\sum_{i=1}^{n}w_{i}\ln\left[\ma\vx-\vb\right]_{i}$
for all $\vx\in P$. Now let $\vx_{*}=\argmin_{\vy\in\R^{d}}p(\vx)$
and suppose that we have $\gamma\geq1$ such that $\left|(\ma\vh)_{i}/(\ma\vx_{*}-\vb)_{i}\right|\leq\gamma\norm{\vh}_{\nabla^{2}p(\vx_{*})}$
for all $\vh\in\R^{d}$. Then the ellipsoid $E=\{\vh\in\R^{d}\,:\,\norm{\vh}_{\nabla^{2}p(\vx_{*})}\leq1\}$
satisfies
\begin{equation}
\vx+\frac{1}{\gamma}E\subset P\subset\vx+\gamma\norm{\vWeight}_{1}E\,.\label{eq:ellipse_rounding_statement}
\end{equation}
\end{lem}
\begin{proof}
For any $\vz\in\frac{1}{\gamma}E$, by the assumptions, we have that
\[
\max_{i\in[n]}\left|\frac{(\ma\vz)_{i}}{(\ma\vx_{*}-\vb)_{i}}\right|\leq\gamma\norm{\vz}_{\hess p(\vx_{*})}\leq1\,.
\]
Consequently, for all $\vz\in\frac{1}{\gamma}E$ we have $\ma(\vx_{*}+\vz)\geq\vb$
and therefore $\vx_{*}+\frac{1}{\gamma}B_{\nabla^{2}p(\vx_{*})}\subset P$. 

To prove the other side of \eqref{eq:ellipse_rounding_statement}
let $\vx\in P$ be arbitrary and let $\vs\defeq\ma\vx-\vb$, $\vs_{*}\defeq\ma\vx_{*}-\vb$,
$\ms\defeq\mDiag(\vs)$, $\ms_{*}=\mDiag(\vx_{*})$, and $\mWeight\defeq\mDiag(\vw)$.
By the optimality conditions on $\vx_{*}$ we have 
\[
\ma\mWeight\ms_{*}^{-1}\onesVec=\grad p(\vx_{*})=\vzero\,.
\]
Consequently, 
\begin{equation}
0=\onesVec^{T}\mWeight\ms_{*}^{-1}\ma\left(\vx_{*}-\vx\right)=\onesVec^{T}\mWeight\ms_{*}^{-1}\left(\vs_{*}-\vs\right)=0\label{eq:a_long_zero}
\end{equation}
and therefore as $\onesVec^{T}\mWeight\ms_{*}^{-1}\vs_{*}=\norm{\vw}_{1}$
we have $\onesVec^{T}\mWeight\ms_{*}^{-1}\vs_{\vx}=\norm{\vWeight}_{1}$
and therefore
\begin{align*}
\sum_{i=1}^{n}w_{i}\frac{[\vs_{*}-\vs]_{i}^{2}}{[\vs_{*}]_{i}^{2}} & =\normFull{\ms_{*}^{-1}\vs_{*}}_{\mWeight}^{2}-2\vs_{*}^{T}\ms_{*}^{-1}\mWeight\ms_{*}^{-1}\vs+\normFull{\ms_{*}^{-1}\vs}_{\mWeight}^{2}\tag{Expanding the quadratic}\\
 & =\normFull{\ms_{*}^{-1}\vs_{\vx}}_{\mWeight}^{2}-\norm{\vWeight}_{1}\tag{Using \ensuremath{\onesVec^{T}\mWeight\ms_{*}^{-1}\vs}=\ensuremath{\norm{\vw}_{1}}}\\
 & \leq\norm{\mWeight\ms_{*}^{-1}\vs_{\vx}}_{1}\norm{\ms_{*}^{-1}\vs_{\vx}}_{\infty}-\norm{\vWeight}_{1}\tag{Cauchy Schwarz}\\
 & =\norm{\vWeight}_{1}\normFull{\ms_{*}^{-1}\left(\vs-\ms_{*}\right)}_{\infty}\tag{Using \ensuremath{\onesVec^{T}\mWeight\ms_{*}^{-1}\vs}=\ensuremath{\norm{\vw}_{1}}}\\
 & \leq\norm{\vWeight}_{1}\cdot\gamma\norm{\vx-\vx_{*}}_{\hess p(\vx_{*})}\tag{Assumption}
\end{align*}
Now, $\hess p(\vx^{*})=\ma^{T}\ms_{*}^{-1}\mWeight\ms_{*}^{-1}\ma$
and therefore $\sum_{i=1}^{n}w_{i}\frac{[\vs_{*}-\vs]_{i}^{2}}{[\vs_{*}]_{i}^{2}}=\norm{\vx-\vx_{*}}_{\nabla^{2}p(\vx_{*})}^{2}$.
Dividing both sides of the above equation by $\norm{\vx-\vx_{*}}_{\nabla^{2}p(\vx_{*})}$
yields that $\norm{\vx-\vx_{*}}_{\nabla^{2}p(\vx_{*})}\leq\gamma\norm{\vWeight}_{1}$
yielding the right hand side of \eqref{eq:ellipse_rounding_statement}
and completing the proof.\end{proof}
\begin{thm}
\label{thm:computing_rounding} Let $\ma\in\R^{n\times d}$ for $d\leq n$
and suppose we have an initial point $\vx_{0}\in\R^{d}$ such that
$\ma\vx_{0}\geq\vb$. Then, we can find an ellipsoid $E$ such that
$E\subset\{\vx\in\R^{d}\,:\,\ma\vx\geq\vb\}\subset100d\cdot E$ in
time $\tilde{O}(\sqrt{d}\left(\nnz(\ma)+d^{2}\right)\log(U))$ where
$U=\max\left\{ \max_{\ma\vx\geq\vb}\norm{\ma\vx-\vb}_{\infty},\normFull{\frac{1}{\ma\vx_{0}-\vb}}_{\infty}\right\} $. \end{thm}
\begin{proof}
Given any interior point $\vx_{0}$ such that $\ma\vx_{0}>\vb$ in
\cite{lsInteriorPoint} we showed how to find a weight $\vWeight$
such that $\norm{\vWeight}_{1}\leq3d$ and for any $\vh\in\R^{d}$
we have 
\begin{equation}
\max_{i\in[n]}\left|\frac{[\ma\vh]_{i}}{[\ma\vx_{*}-\vb]_{i}}\right|\leq3\norm{\vh}_{\nabla^{2}p(\vx_{*})}\label{eq:w_gamma_condition}
\end{equation}
where $p(\vx)=-\sum_{i=1}^{n}w_{i}\ln\left[\ma\vx-\vb\right]_{i}$
and $\vx_{*}=\argmin_{\ma\vy\geq\vb}p(\vy)$. Furthermore, we showed
how to find $\vx\in\R^{d}$ such that 
\begin{equation}
\sum_{i=1}^{n}w_{i}\frac{[\ma(\vx-\vx_{*})]_{i}^{2}}{(\ma\vx_{*}-\vb)_{i}^{2}}\leq\frac{1}{100}\,.\label{eq:apx_E}
\end{equation}
As discussed in Theorem~\ref{thm:LPSolve}, this can be done in time
$\tilde{O}(\sqrt{d}\left(\nnz(\ma)+d^{2}\right)\log(U))$. Consequently,
we can use Lemma~\ref{lem:John_ellipsoid_gurantee} with $\gamma=3$
and $\norm{\vWeight}_{1}=3d$ and this gives us an ellipsoid $E$
such that
\[
\vx_{*}+\frac{1}{3}E\subset P\subset\vx_{*}+9dE
\]
where $E=\{\vh\in\R^{d}\,:\,\norm{\vh}_{\nabla^{2}p(\vx_{*})}\leq1\}$. 

Now, we show how to approximate $E$ using $\vx$. \eqref{eq:apx_E}
shows that that $\norm{\vx-\vx_{*}}_{\nabla^{2}p(\vx_{*})}\leq\frac{1}{10}$,
therefore, we have $\vx-\vx_{*}\in\frac{1}{10}E$ and hence
\[
\vx+\left(\frac{1}{3}-\frac{1}{10}\right)E\subset\vx_{*}+\frac{1}{3}E\subset P\subset\vx_{*}+9dE\subset\vx+\left(9d+\frac{1}{10}\right)E.
\]
Now, using \eqref{eq:apx_E} and \eqref{eq:w_gamma_condition}, we
have $\left|\frac{\ma(\vx-\vx_{*})}{\ma\vx_{*}-\vb}\right|_{i}\leq\frac{3}{10}$,
we have 
\[
\frac{1}{(1+\frac{3}{10})^{2}}\nabla^{2}p(\vx)\preceq\nabla^{2}p(\vx_{*})=\ma^{T}\ms_{*}^{-1}\mWeight\ms_{*}^{-1}\ma\preceq\frac{1}{(1-\frac{3}{10})^{2}}\nabla^{2}p(\vx)\,.
\]
Therefore, we have
\[
\vx+\frac{\frac{1}{3}-\frac{1}{10}}{1+\frac{3}{10}}E'\subset P\subset\vx+\frac{9d+\frac{1}{10}}{1-\frac{3}{10}}E'
\]
where $E'=\{\vh\in\R^{d}\,:\,\norm{\vh}_{\nabla^{2}p(\vx)}\leq1\}$.
After rescaling the ellipsoid, we have the result.
\end{proof}

\subsection{Multicommodity Flow \label{sub:app:multicommodity_flow}}

Here we show how our algorithm can be used to improve the running
time for solving multicommodity flow. Note that the result presented
here is meant primarily to illustrate our approach, we believe it
can be further improved using the techniques in \cite{kapoor1996speeding,lsMaxflow}.

For simplicity, we focus on the maximum concurrent flow problem. In
this problem, we are given a graph $G=(V,E)$ $k$ source sink pairs
$(s_{i},t_{i})\in\R^{V\times V}$ , and capacities $\vc\in\R^{E}$
and wish to compute the maximum $\alpha\in\R$ such that we can simultaneously
for all $i\in[k]$ route $\alpha$ unit of flow $f_{i}\in\R^{E}$
between $s_{i}$ and $t_{i}$ while maintaining the capacity constraint
$\sum_{i=1}^{k}\left|f_{i}(e)\right|\leq c(e)$ for all $e\in E$.
There are no combinatorial algorithms known and there are multiple
algorithms to compute a $(1-\epsilon)$ optimal flow in time polynomial
in $|E|$, $|V|$ and $\epsilon^{-1}$ \cite{fleischer2000approximating,karakostas2002faster,garg2007faster,madry2010faster}.
In this regime, the fastest algorithm for directed graphs takes $O((|E|+k)|V|\epsilon^{-2}\log U)$
\cite{madry2010faster} and the fastest algorithm for undirected graphs
takes $\tilde{O}(|E|^{1+o(1)}k^{2}\epsilon^{-2}\log^{O(1)}U)$ \cite{lee2014linearmaxflow}
where $U$ is the maximum capacity. For linear convergence, the previous
best algorithm is a specialized interior point method that takes time
$O(\sqrt{|E|k}|V|^{2}k^{2}\log(\epsilon^{-1}))$ \cite{kapoor1996speeding}.

To solve the concurrent multicommodity flow problem we use the following
linear program
\begin{eqnarray*}
\max & \alpha\\
g(e) & = & \sum f_{i}(e)\text{ for all edges }e,\\
\sum_{v\in V}f_{i}(u,v) & = & 0\text{ for all vertices }u\notin\{s_{i},t_{i}\},\\
\sum_{v\in V}f_{i}(s_{i},v) & = & \alpha\text{ for all }i,\\
c(e)\geq & f_{i}(e) & \geq0\enspace\text{for all edges \ensuremath{e},}\\
c(e)\geq & g(e) & \geq0\enspace\text{for all edges \ensuremath{e}}.
\end{eqnarray*}
Note that there are $O(k|E|)$ variables, $O(k|V|+|E|)$ equality
constraints, and $O(k|E|)$ non-zeroes in the constraint matrix. Furthermore,
it is easy to find an initial point by computing a shortest path for
each $s_{i}$ and $t_{i}$ and sending a small amount of flow along
that path. Also, given any almost feasible flow, one can make it feasible
by scaling and send excess flow at every vertex back to $s_{i}$ along
some spanning tree. Therefore, Theorem~\ref{thm:LPSolve} gives an
algorithm that takes $\tilde{O}(\sqrt{|E|+k|V|}\left(k|E|+(|E|+k|V|)^{2}\right)\log(U/\epsilon))=\otilde\left((|E|+k|V|)^{2.5}\log(U/\epsilon)\right)$
time. Note that this is faster than the previous best algorithm when
$k\geq(|E|/|V|)^{0.8}$.

\subsection{Minimum Cost Flow\label{sub:app:mincost_flow}}

The minimum cost flow problem and the more specific, maximum flow
problem, are two of the most well studied problems in combinatorial
optimization \cite{schrijver2003combinatorial}. Many techniques have
been developed for solving this problem. Yet, our result matches the
fastest algorithm for solving these problems on dense graphs \cite{lsMaxflow}
without using any combinatorial structure of this problem, in particular,
Laplacian system solvers. We emphasize that this result is not a running
time improvement, rather just a demonstration of the power of our
result and an interesting statement on efficient maximum flow algorithms.
\begin{cor}
There is an $\otilde(|V|^{2.5}\log^{O(1)}\left(U\right))$ time algorithm
to compute an exact minimum cost maximum flow for weighted directed
graphs with $|V|$ vertices, $|E|$ edges and integer capacities and
integer cost at most $U$.\end{cor}
\begin{proof}
The proof is same as \cite[ArXiv v2, Thm 34,35]{lsMaxflow} except
that we use Theorem~\ref{thm:LPSolve} to solve the linear program.
The proof essentially writes the minimum cost flow problem into a
linear program with an explicit interior point and shows how to round
an approximately optimal solution to a vertex of the polytope. To
perform this rounding, we need to a fractional solution with error
less than $O(\frac{1}{\poly(|V|U)})$ and which yields the $\log(U)$
term in the running time.
\end{proof}

\subsection{Convex Problems\label{sub:app:convex_problem}}

Many problems in convex optimization can be efficiently reduced to
the problem of finding a point in a convex set $K$ given a separation
oracle. Recall that given a point $\vx$, the separation oracle either
outputs that $\vx$ is in $K$ or outputs a separating hyperplane
that separates the input point $\vx$ and the convex set $K$. In
\cite{lee2015convex}, they showed that how to make use of our fast
algorithms for inverse maintenance problem under the $K$ stability
assumption to obtain the following improved running time for this
fundamental problem:
\begin{thm}[\cite{lee2015convex}]
Given a non-empty convex set $K\subseteq\R^{d}$ that is contained
in a box of radius $R$, i.e. $\max_{x\in K}\norm x_{\infty}\leq R$.
We are also given a separation oracle for $K$ that takes $O(\mathcal{T})$
time for each call. For any $0<\epsilon<R$, we can either finds $\vx\in K$
or proves that $K$ does not contains a ball with radius $\epsilon$
in time $O(d\mathcal{T}\log(dR/\epsilon)+d^{3}\log^{O(1)}(dR/\epsilon))$.\end{thm}
\begin{rem}
\cite{lee2015convex} uses Theorem~\ref{thm:iteration_solver_2}
to solve the linear systems involved in their cutting plane method.
However, their inverse maintenance problem satisfies the $\ell^{2}$
stability assumption with $1$ rows addition and removal. Furthermore,
the $\ma$ involved has only $O(d)$ many rows. Therefore, we believe
a simple variant of \cite{vaidya1989speeding} may also suffice for
that paper.
\end{rem}

\section{Open Problem: Sampling from a Polytope \label{sub:app:sample}}

Sampling a random point in convex sets is a fundamental problem in
convex geometry with numerous applications in optimization, counting,
learning and rounding \cite{vempala2010recent}. Here consider a typical
case where the convex set is a polytope that is explicitly given as
$\{\vx\in\R^{d}\,:\,\ma\vx\geq\vb\}$ for $\ma\in\R^{n\times d}.$
The current fastest algorithm for this setting is Hit-and-Run \cite{lovasz2006fast}
and Dikin walk \cite{kannan2012random}. 

Given an initial random point in the set, Hit-and-Run takes $O^{*}(d^{3})$
iterations and each iteration takes time $O^{*}(\nnz(A))$ while Dikin
walk takes $O^{*}(nd)$ iterations and each iteration takes time $O^{*}(nd^{\omega-1})$
where the $O^{*}$ notation omits the dependence on error parameters
and logarithmic terms. Each iteration of Dikin walk is expensive because
it solves a linear system to obtain the next point and computes determinants
to implement an importance sampling scheme. The linear systems can
be solved in amortized cost $O^{*}(\nnz(\ma)+d^{2})$ by the inverse
maintenance machinery presented in this paper. Unfortunately it is
not known how to use this machinery to efficiently compute the determinant
to sufficient accuracy to suffice for this method. 

We leave it as an open problem whether it is possible to circumvent
this issue and improve the running time of a method like the Dikin
walk. In an older version of this paper, we mistakenly claimed an
improved running time for Dikin walk by noting solely the improved
running time for linear system solving and ignoring the determinant
computation. We thank Hariharan Narayanan for pointing out this mistake.

%% file: sec_iteration_appendix.tex
\section{Relationships Between $\mathcal{T}$-time Linear Solver and Inverse
Matrix\label{sec:relationships_linear_spectral}}

In this section, we prove relationships between linear $\time$-time
solvers and inverse matrices. In Lemma~\ref{lem:approx_implies_solver}
we show that obtaining a spectral approximation to the inverse of
a matrix suffices to obtain a linear solver. In Lemma~\ref{lem:solver_implies_approx},
we show how a linear solver of a matrix yields a spectral approximation
to the inverse of that matrix.
\begin{lem}
\label{lem:approx_implies_solver} Given a PD matrix $\mm$ and a
symmetric matrix $\mn\approx_{O(1)}\mm^{-1}$ that can be applied
to a vector in $O(\time)$ time we have a linear $(\nnz(\mm)+\time)$-time
solver $\solver$ of $\mm$. \end{lem}
\begin{proof}
Since $\mn\approx_{O(1)}\mm^{-1}$ there is a constant $L$ such that
$\frac{1}{L}\mn^{-1}\preceq\mm\preceq L\mn$. Consider the algorithm
$\solver(\vb,\epsilon)\defeq\mq_{\epsilon}\vb$ where 
\[
\mq_{\epsilon}\defeq\frac{1}{L}\mn\sum_{k=0}^{z_{\varepsilon}}\left(\iMatrix-\frac{\mm\mn}{L}\right)^{k}
\]
where $z_{\epsilon}=\frac{1}{2}\log(\epsilon L^{-4})/\log(1-L^{-2})$.
This is the linear operator corresponding to performing $z_{\epsilon}$
steps of preconditioned gradient descent to solve the linear system
in $\mm$. Clearly, applying $\mq_{\epsilon}$ can be done in time
$O((\nnz(\ma)+\time)z_{e})=O((\nnz(\ma)+\time)\log(\epsilon^{-1}))$.
Therefore, all that remains is to show that $\solver(\vb,\epsilon)$
is a solver for $\mm$.

Note that we can rewrite $\mq_{\epsilon}$ equivalently as 
\[
\mq_{\epsilon}=\frac{1}{L}\mn^{1/2}\sum_{k=0}^{z_{\varepsilon}}\left(\iMatrix-\frac{1}{L}\mn^{1/2}\mm\mn^{1/2}\right)^{k}\mn^{1/2}
\]
and hence it is symmetric. Furthermore, since clearly $\frac{1}{L^{2}}\iMatrix\specLeq\frac{1}{L}\mn^{1/2}\mm\mn^{1/2}\specLeq\iMatrix$
we have that
\begin{align*}
\mm^{-1}-\mq_{\epsilon} & =\frac{1}{L}\mn^{1/2}\left(\iMatrix-\left(\iMatrix-\frac{1}{L}\mn^{1/2}\mm\mn^{1/2}\right)\right)^{-1}\mn^{1/2}-\frac{1}{L}\mn^{1/2}\sum_{k=0}^{z_{\varepsilon}}\left(\iMatrix-\frac{1}{L}\mn^{1/2}\mm\mn^{1/2}\right)^{k}\mn^{1/2}\\
 & =\frac{1}{L}\mn^{1/2}\sum_{k=z_{\varepsilon}+1}^{\infty}\left(\iMatrix-\frac{1}{L}\mn^{1/2}\mm\mn^{1/2}\right)^{k}\mn^{1/2}\enspace.
\end{align*}
Using the above two inequalities, we have that 
\begin{eqnarray*}
\norm{\solver(\vb,\epsilon)-\mm^{-1}\vb}_{\mm}^{2} & = & \normFull{\frac{1}{L}\mn^{1/2}\sum_{k=z_{\varepsilon}+1}^{\infty}\left(\iMatrix-\frac{1}{L}\mn^{1/2}\mm\mn^{1/2}\right)^{k}\mn^{1/2}\vb}_{\mm}^{2}\\
 & \leq & \frac{1}{L}\normFull{\sum_{k=z_{\varepsilon}+1}^{\infty}\left(\iMatrix-\frac{1}{L}\mn^{1/2}\mm\mn^{1/2}\right)^{k}\mn^{1/2}\vb}_{2}^{2}
\end{eqnarray*}
and since using that $\mZero\specLeq\iMatrix-\frac{1}{L}\mn^{1/2}\mm\mn^{1/2}\specLeq\left(1-\frac{1}{L^{2}}\right)\iMatrix$,
we have
\begin{eqnarray*}
\mvar 0\preceq\sum_{k=z_{\varepsilon}+1}^{\infty}\left(\iMatrix-\frac{1}{L}\mn^{1/2}\mm\mn^{1/2}\right)^{k} & \preceq & L^{2}\left(1-\frac{1}{L^{2}}\right)^{z_{\varepsilon}}\iMatrix\enspace.
\end{eqnarray*}
Consequently
\[
\norm{\solver(\vb,\varepsilon)-\mm^{-1}\vb}_{\mm}^{2}\leq L^{3}\left(1-\frac{1}{L^{2}}\right)^{2z_{\varepsilon}}\normFull{\mn^{1/2}\vb}_{2}^{2}\leq L^{4}\left(1-\frac{1}{L^{2}}\right)^{2z_{\varepsilon}}\normFull{\mm^{-1}\vb}_{\mm}^{2}\,.
\]
Thus, we see that $z_{\epsilon}$ was chosen precisely to complete
the proof.\end{proof}
\begin{lem}
\label{lem:solver_implies_approx} Let $\solver$ be a linear solver
of $\mm$ such that $\solver(\vb,\epsilon)=\mq_{\epsilon}\vb$ for
$\epsilon\in[0,0.1]$. Then, we have
\[
\mq_{\epsilon}^{T}\mm\mq_{\epsilon}\approx_{4\sqrt{\epsilon}}\mm^{-1}\text{ and }\mq_{\epsilon}\mm\mq_{\epsilon}^{T}\approx_{4\sqrt{\epsilon}}\mm^{-1}
\]
\end{lem}
\begin{proof}
By the definition of $\solver$ and standard inequalities, we have
that for all $\vb$,
\begin{eqnarray*}
\norm{\mathcal{\solver}(\vb,\epsilon)}_{\mm}^{2} & \leq & \left(1+\frac{1}{\sqrt{\epsilon}}\right)\norm{\mathcal{S}(\vb,\varepsilon)-\mm^{-1}\vb}_{\mm}^{2}+\left(1+\sqrt{\epsilon}\right)\norm{\mm^{-1}\vb}_{\mm}^{2}\\
 & \leq & \left(\epsilon+\sqrt{\epsilon}+1+\sqrt{\epsilon}\right)\norm{\mm^{-1}\vb}_{\mm}^{2}\leq(1+3\sqrt{\epsilon})\norm{\mm^{-1}\vb}_{\mm}^{2}
\end{eqnarray*}
On the another hand, we have that for all $\vb$
\begin{eqnarray*}
\norm{\mm^{-1}\vb}_{\mm}^{2} & \leq & \left(1+\frac{1}{\sqrt{\epsilon}}\right)\norm{\mathcal{S}(\vb,\epsilon)-\mm^{-1}\vb}_{\mm}^{2}+\left(1+\sqrt{\varepsilon}\right)\norm{\mathcal{S}(\vb,\epsilon)}_{\mm}^{2}\\
 & \leq & 2\sqrt{\epsilon}\norm{\mm^{-1}\vb}_{\mm}^{2}+\left(1+\sqrt{\epsilon}\right)\norm{\mathcal{S}(\vb,\epsilon)}_{\mm}^{2}.
\end{eqnarray*}
Combining these two inequalities and using the definition of $\solver$,
we see that for all $\vb$ we have
\[
\left(\frac{1-2\sqrt{\epsilon}}{1+\sqrt{\epsilon}}\right)\vb^{T}\mm^{-1}\vb\leq\vb^{T}\mq_{\epsilon}^{T}\mm\mq_{\epsilon}\vb\leq\left(1+3\sqrt{\epsilon}\right)\vb^{T}\mm^{-1}\vb
\]
Using a Taylor expansion of $e$ and the definition of $\approx$
then yields $\mq_{\epsilon}^{T}\mm\mq_{\epsilon}\approx_{4\sqrt{\epsilon}}\mm^{-1}$.
In other words, all eigenvalues of $\left(\mm^{1/2}\mq_{\varepsilon}\mm^{1/2}\right)^{T}\left(\mm^{1/2}\mq_{\varepsilon}\mm^{1/2}\right)$
lies between $e^{-4\sqrt{\varepsilon}}$ and $e^{4\sqrt{\varepsilon}}$.
In general, for any square matrix $\mb$, the set of eigenvalues of
$\mb^{T}\mb$ equals to the set of eigenvalues of $\mb\mb^{T}$ because
of the SVD decomposition. Hence, all eigenvalues of $\left(\mm^{1/2}\mq_{\varepsilon}\mm^{1/2}\right)\left(\mm^{1/2}\mq_{\varepsilon}\mm^{1/2}\right)^{T}$
also lies in the same region. This proves $\mq_{\epsilon}\mm\mq_{\epsilon}^{T}\approx_{4\sqrt{\epsilon}}\mm^{-1}$.
\end{proof}

\section{Remarks for Figure \ref{fig:run_time}}

\label{sec:figure_remarks}

Figure~\ref{fig:run_time} shows the previous fastest algorithm for
linear programming $\min_{\ma\vx\geq\vb}\vc^{T}\vx$ where $\ma\in\R^{n\times d}$
. The running time described in the figure comes from the following
algorithms:
\begin{enumerate}
\item $\tilde{O}\left(\sqrt{n}z+\sqrt{n}d^{2}+n^{1.34}d^{1.15}\right)$
is achieved by the interior point method of Vaidya \cite{vaidya1989speeding}. 
\item $\tilde{O}\left(\sqrt{n}z+nd^{1.38}+n^{0.69}d^{2}\right)$ is achieved
by using the Karmarkar acceleration scheme \cite{karmarkar1984new}
on the short step path following method. See \cite{nesterov1991acceleration}
for the details.
\item $\tilde{O}(\sqrt{n}(z+d^{2.38}))$ is achieved by using the currently
best linear system solver \cite{nelson2012osnap,li2012iterative,gall2014powers,Cohen2014}
on the short step path following method \cite{renegar1988polynomial}.
\item $\tilde{O}(d(z+d^{2.38}))$ is achieved by the cutting plane method
of Vaidya \cite{vaidya1996new}.\end{enumerate}

%% file: iteration_main.bbl
\begin{thebibliography}{10}

\bibitem{chin2013runtime}
Hui~Han Chin, Aleksander Madry, Gary~L Miller, and Richard Peng.
\newblock Runtime guarantees for regression problems.
\newblock In {\em Proceedings of the 4th conference on Innovations in
  Theoretical Computer Science}, pages 269--282. ACM, 2013.

\bibitem{clarkson2013low}
Kenneth~L Clarkson and David~P Woodruff.
\newblock Low rank approximation and regression in input sparsity time.
\newblock In {\em Proceedings of the 45th annual ACM symposium on Symposium on
  theory of computing}, pages 81--90. ACM, 2013.

\bibitem{Cohen2014}
Michael~B. Cohen, Yin~Tat Lee, Cameron Musco, Christopher Musco, Richard Peng,
  and Aaron Sidford.
\newblock Uniform sampling for matrix approximation.
\newblock {\em CoRR}, abs/1408.5099, 2014.

\bibitem{cohen2014ell_p}
Michael~B Cohen and Richard Peng.
\newblock Lp row sampling by lewis weights.
\newblock {\em arXiv preprint arXiv:1412.0588}, 2014.

\bibitem{drineas2006subspace}
Petros Drineas, Michael~W Mahoney, and S~Muthukrishnan.
\newblock Subspace sampling and relative-error matrix approximation:
  Column-based methods.
\newblock In {\em Approximation, Randomization, and Combinatorial Optimization.
  Algorithms and Techniques}, pages 316--326. Springer, 2006.

\bibitem{fleischer2000approximating}
Lisa~K Fleischer.
\newblock Approximating fractional multicommodity flow independent of the
  number of commodities.
\newblock {\em SIAM Journal on Discrete Mathematics}, 13(4):505--520, 2000.

\bibitem{gall2014powers}
Fran{\c{c}}ois~Le Gall.
\newblock Powers of tensors and fast matrix multiplication.
\newblock {\em arXiv preprint arXiv:1401.7714}, 2014.

\bibitem{garg2007faster}
Naveen Garg and Jochen Koenemann.
\newblock Faster and simpler algorithms for multicommodity flow and other
  fractional packing problems.
\newblock {\em SIAM Journal on Computing}, 37(2):630--652, 2007.

\bibitem{kannan2012random}
Ravindran Kannan and Hariharan Narayanan.
\newblock Random walks on polytopes and an affine interior point method for
  linear programming.
\newblock {\em Mathematics of Operations Research}, 37(1):1--20, 2012.

\bibitem{kapoor1996speeding}
Sanjiv Kapoor and Pravin~M Vaidya.
\newblock Speeding up karmarkar's algorithm for multicommodity flows.
\newblock {\em Mathematical programming}, 73(1):111--127, 1996.

\bibitem{karakostas2002faster}
George Karakostas.
\newblock Faster approximation schemes for fractional multicommodity flow
  problems.
\newblock {\em ACM Transactions on Algorithms (TALG)}, 4(1):13, 2008.

\bibitem{karmarkar1984new}
Narendra Karmarkar.
\newblock A new polynomial-time algorithm for linear programming.
\newblock In {\em Proceedings of the sixteenth annual ACM symposium on Theory
  of computing}, pages 302--311. ACM, 1984.

\bibitem{lee2014linearmaxflow}
Jonathan~A Kelner, Yin~Tat Lee, Lorenzo Orecchia, and Aaron Sidford.
\newblock An almost-linear-time algorithm for approximate max flow in
  undirected graphs, and its multicommodity generalizations.
\newblock In {\em SODA}, pages 217--226. SIAM, 2014.

\bibitem{khachiyan1996rounding}
Leonid~G Khachiyan.
\newblock Rounding of polytopes in the real number model of computation.
\newblock {\em Mathematics of Operations Research}, 21(2):307--320, 1996.

\bibitem{kumar2005minimum}
Piyush Kumar and E~Alper Yildirim.
\newblock Minimum-volume enclosing ellipsoids and core sets.
\newblock {\em Journal of Optimization Theory and Applications}, 126(1):1--21,
  2005.

\bibitem{laurent2000adaptive}
B{\'e}atrice Laurent and Pascal Massart.
\newblock Adaptive estimation of a quadratic functional by model selection.
\newblock {\em Annals of Statistics}, pages 1302--1338, 2000.

\bibitem{lsInteriorPoint}
Yin~Tat Lee and Aaron Sidford.
\newblock Path finding i: Solving linear programs with $\backslash$\~{} o
  (sqrt(rank)) linear system solves.
\newblock {\em arXiv preprint arXiv:1312.6677}, 2013.

\bibitem{lsMaxflow}
Yin~Tat Lee and Aaron Sidford.
\newblock Path finding ii: An$\backslash$\~{} o (m sqrt (n)) algorithm for the
  minimum cost flow problem.
\newblock {\em arXiv preprint arXiv:1312.6713}, 2013.

\bibitem{leeS14}
Yin~Tat Lee and Aaron Sidford.
\newblock Path-finding methods for linear programming : Solving linear programs
  in \~o(sqrt(rank)) iterations and faster algorithms for maximum flow.
\newblock In {\em 55th Annual {IEEE} Symposium on Foundations of Computer
  Science, {FOCS} 2014, 18-21 October, 2014, Philadelphia, PA, {USA}}, pages
  424--433, 2014.

\bibitem{lee2015convex}
Yin~Tat Lee, Aaron Sidford, and Sam~Chiu{-}wai Wong.
\newblock A faster cutting plane method and its implications for combinatorial
  and convex optimization.
\newblock In {\em Foundations of Computer Science (FOCS), 2015 IEEE 56th Annual
  Symposium on}. IEEE, 2015.

\bibitem{li2012iterative}
Mu~Li, Gary~L Miller, and Richard Peng.
\newblock Iterative row sampling.
\newblock 2012.

\bibitem{lovasz2006fast}
L{\'a}szl{\'o} Lov{\'a}sz and Santosh Vempala.
\newblock Fast algorithms for logconcave functions: Sampling, rounding,
  integration and optimization.
\newblock In {\em Foundations of Computer Science, 2006. FOCS'06. 47th Annual
  IEEE Symposium on}, pages 57--68. IEEE, 2006.

\bibitem{madry2010faster}
Aleksander Madry.
\newblock Faster approximation schemes for fractional multicommodity flow
  problems via dynamic graph algorithms.
\newblock In {\em Proceedings of the forty-second ACM symposium on Theory of
  computing}, pages 121--130. ACM, 2010.

\bibitem{mahoney11survey}
Michael~W. Mahoney.
\newblock Randomized algorithms for matrices and data.
\newblock {\em Foundations and Trends in Machine Learning}, 3(2):123--224,
  2011.

\bibitem{mahoney2012fast}
Michael~W Mahoney, Petros Drineas, Malik Magdon-Ismail, and David~P Woodruff.
\newblock Fast approximation of matrix coherence and statistical leverage.
\newblock In {\em ICML}, 2012.

\bibitem{nelson2012osnap}
Jelani Nelson and Huy~L Nguy{\^e}n.
\newblock Osnap: Faster numerical linear algebra algorithms via sparser
  subspace embeddings.
\newblock {\em arXiv preprint arXiv:1211.1002}, 2012.

\bibitem{nesterov1989self}
Yu~Nesterov and Arkadi Nemirovskiy.
\newblock {\em Self-concordant functions and polynomial-time methods in convex
  programming}.
\newblock USSR Academy of Sciences, Central Economic \& Mathematic Institute,
  1989.

\bibitem{nesterov1991acceleration}
Yu~Nesterov and A~Nemirovsky.
\newblock Acceleration and parallelization of the path-following interior point
  method for a linearly constrained convex quadratic problem.
\newblock {\em SIAM Journal on Optimization}, 1(4):548--564, 1991.

\bibitem{Nesterov1994}
Yurii Nesterov and Arkadii~Semenovich Nemirovskii.
\newblock {\em Interior-point polynomial algorithms in convex programming},
  volume~13.
\newblock Society for Industrial and Applied Mathematics, 1994.

\bibitem{renegar1988polynomial}
James Renegar.
\newblock A polynomial-time algorithm, based on newton's method, for linear
  programming.
\newblock {\em Mathematical Programming}, 40(1-3):59--93, 1988.

\bibitem{schrijver2003combinatorial}
Alexander Schrijver.
\newblock {\em Combinatorial optimization: polyhedra and efficiency},
  volume~24.
\newblock Springer, 2003.

\bibitem{spielmanS08sparsRes}
Daniel~A Spielman and Nikhil Srivastava.
\newblock Graph sparsification by effective resistances.
\newblock {\em SIAM Journal on Computing}, 40(6):1913--1926, 2011.

\bibitem{vaidya1989speeding}
Pravin~M Vaidya.
\newblock Speeding-up linear programming using fast matrix multiplication.
\newblock In {\em Foundations of Computer Science, 1989., 30th Annual Symposium
  on}, pages 332--337. IEEE, 1989.

\bibitem{vaidya1996new}
Pravin~M Vaidya.
\newblock A new algorithm for minimizing convex functions over convex sets.
\newblock {\em Mathematical Programming}, 73(3):291--341, 1996.

\bibitem{vempala2010recent}
Santosh~S Vempala.
\newblock Recent progress and open problems in algorithmic convex geometry.
\newblock In {\em LIPIcs-Leibniz International Proceedings in Informatics},
  volume~8. Schloss Dagstuhl-Leibniz-Zentrum fuer Informatik, 2010.

\bibitem{williams2012matrixmult}
Virginia~Vassilevska Williams.
\newblock Multiplying matrices faster than coppersmith-winograd.
\newblock In {\em Proceedings of the forty-fourth annual ACM symposium on
  Theory of computing}, pages 887--898. ACM, 2012.

\end{thebibliography}
